\documentclass[11pt,a4paper]{amsart}

\usepackage[T1]{fontenc}
\usepackage[english]{babel}
\usepackage{mlmodern}
\usepackage{microtype}

\usepackage{p-notation}
\usepackage{enumitem}
\usepackage{booktabs}

\usepackage{csquotes}
\usepackage[
  backend=biber,
  style=numeric-comp,
  sorting=nyt
]{biblatex}
\usepackage[dvipsnames,svgnames,x11names]{xcolor}
\usepackage[normalem]{ulem}

\newif\ifworkingdraft
\workingdraftfalse

\ifworkingdraft
  \usepackage[
    colorinlistoftodos,
    prependcaption,
    textsize=footnotesize
  ]{todonotes}
  \newcommand{\add}[2][]{%
    \todo[linecolor=BrickRed,backgroundcolor=Salmon!30,
      bordercolor=BrickRed,#1]{#2}%
  }
  \newcommand{\dopawel}[2][]{%
    \todo[linecolor=RoyalBlue,backgroundcolor=SkyBlue!30,
      bordercolor=RoyalBlue,#1]{#2}%
  }
  \newcommand{\doauthortwo}[2][]{%
    \todo[linecolor=Purple,backgroundcolor=Orchid!30,
      bordercolor=Purple,#1]{#2}%
  }
  \newcommand{\question}[2][]{%
    \todo[linecolor=Orange,backgroundcolor=Yellow!30,
      bordercolor=Orange,#1]{[Q] #2}%
  }
  \newcommand{\response}[2][]{%
    \todo[linecolor=Teal,backgroundcolor=Cyan!30,
      bordercolor=Teal,#1]{[R] #2}%
  }
  \newcommand{\important}[2][]{%
    \todo[linecolor=Red!70!black,backgroundcolor=Red!20,
      bordercolor=Red!70!black,size=\normalsize,#1]{\textbf{!} #2}%
  }

  \newcommand{\deletion}[1]{\textcolor{Red}{\sout{#1}}}
  
\else
  \usepackage[disable]{todonotes}

  \newcommand{\add}[2][]{}
  \newcommand{\dopawel}[2][]{}
  \newcommand{\doauthortwo}[2][]{}
  \newcommand{\question}[2][]{}
  \newcommand{\response}[2][]{}
  \newcommand{\important}[2][]{}

  \newcommand{\deletion}[1]{}
  
\fi

\usepackage{array,tabularx}
\usepackage{graphicx}
\usepackage{tikz}
\usetikzlibrary{arrows.meta,positioning,fit,backgrounds,calc}
\usepackage{needspace}
\usepackage{placeins}
\usepackage{etoolbox}
\usepackage[margin=2.7cm]{geometry}

\usepackage{zref-clever}
\usepackage[hidelinks,pdfusetitle]{hyperref}
\hypersetup{
  pdftitle={Medvedev Logic Is Not Decidable. It Is Pi01-Complete. Who Would Have Guessed?},
  pdfauthor={Pawel Pawlowski},
  pdfsubject={Version 8, 9 September 2026}
}

\zcsetup{
  nameinlink = true,
  abbrev = false,
  cap = true
}

\theoremstyle{plain}
\newtheorem{theorem}{Theorem}[section]
\newtheorem{lemma}[theorem]{Lemma}

\newtheorem{fact}[theorem]{Fact}
\newtheorem{corollary}[theorem]{Corollary}

\newtheorem{proposition}[theorem]{Proposition}

\theoremstyle{definition}
\newtheorem{definition}[theorem]{Definition}

\theoremstyle{remark}

\zcRefTypeSetup{axiom}{
  Name-sg = Axiom,
  name-sg = axiom,
  Name-pl = Axioms,
  name-pl = axioms
}
\zcRefTypeSetup{fact}{
  Name-sg = Fact,
  name-sg = fact,
  Name-pl = Facts,
  name-pl = facts
}
\zcRefTypeSetup{conjecture}{
  Name-sg = Conjecture,
  name-sg = conjecture,
  Name-pl = Conjectures,
  name-pl = conjectures
}
\zcRefTypeSetup{convention}{
  Name-sg = Convention,
  name-sg = convention,
  Name-pl = Conventions,
  name-pl = conventions
}

\AddToHook{env/theorem/begin}{\zcsetup{countertype={theorem=theorem}}}
\AddToHook{env/lemma/begin}{\zcsetup{countertype={theorem=lemma}}}
\AddToHook{env/axiom/begin}{\zcsetup{countertype={theorem=axiom}}}
\AddToHook{env/fact/begin}{\zcsetup{countertype={theorem=fact}}}
\AddToHook{env/corollary/begin}{\zcsetup{countertype={theorem=corollary}}}
\AddToHook{env/conjecture/begin}{\zcsetup{countertype={theorem=conjecture}}}
\AddToHook{env/proposition/begin}{\zcsetup{countertype={theorem=proposition}}}
\AddToHook{env/definition/begin}{\zcsetup{countertype={theorem=definition}}}
\AddToHook{env/example/begin}{\zcsetup{countertype={theorem=example}}}
\AddToHook{env/convention/begin}{\zcsetup{countertype={theorem=convention}}}
\AddToHook{env/remark/begin}{\zcsetup{countertype={theorem=remark}}}

\newcommand{\ML}{\mathrm{ML}}
\newcommand{\up}{\mathord\uparrow}

\renewcommand{\D}{\mathcal D}
\newcommand{\tr}{\mathrm t}
\newcommand{\WM}{\mathtt{WM}}
\newcommand{\role}{\operatorname{role}}
\newcommand{\Resp}{\operatorname{Resp}}
\newcommand{\PowPlus}[1]{\mathcal P^{+}(#1)}
\numberwithin{equation}{section}
\allowdisplaybreaks
\newcommand{\Mid}{\mathsf{Mid}}
\newcommand{\MaxPts}{\mathsf{Max}}
\newcommand{\CPos}{\mathsf{Column}}
\newcommand{\RPos}{\mathsf{Row}}
\newcommand{\CSep}{\mathsf{ColumnSep}}
\newcommand{\RSep}{\mathsf{RowSep}}

\newcommand{\axisname}[1]{
  \ifstrequal{#1}{h}{\mathrm{col}}{
  \ifstrequal{#1}{v}{\mathrm{row}}{#1}}}

\newcommand{\NFL}[1]{\mathsf{NotNext}_{\axisname{#1}}}
\newcommand{\SFL}[1]{\mathsf{NotSame}_{\axisname{#1}}}
\newcommand{\CellL}[1]{\mathsf{Cell}_{#1}}
\newcommand{\SrcL}[2]{\mathsf{Source}^{#1}_{#2}}
\newcommand{\NbrL}[2]{\mathsf{Candidate}^{#1}_{#2}}
\newcommand{\ExpL}[2]{\mathsf{Expected}^{#1}_{#2}}
\newcommand{\StepL}[1]{
  \ifstrequal{#1}{h}{\mathsf{ColStep}}{
  \ifstrequal{#1}{v}{\mathsf{RowStep}}{\mathsf{Step}_{#1}}}}
\newcommand{\KeyL}[1]{\mathsf{Key}_{#1}}

\newcommand{\sa}[1]{\mathsf{same}_{\axisname{#1}}}
\newcommand{\nx}[1]{\mathsf{next}_{\axisname{#1}}}
\newcommand{\tg}[1]{\mathsf{tag}_{\axisname{#1}}}
\newcommand{\keyc}[1]{\kappa_{#1}}

\newcommand{\MaxOf}[1]{\mathrm{Max}(#1)}

\AtBeginEnvironment{definition}{\Needspace{6\baselineskip}}
\AtBeginEnvironment{theorem}{\Needspace{5\baselineskip}}
\AtBeginEnvironment{lemma}{\Needspace{5\baselineskip}}
\AtBeginEnvironment{proposition}{\Needspace{5\baselineskip}}
\AtBeginEnvironment{corollary}{\Needspace{5\baselineskip}}
\AtBeginEnvironment{fact}{\Needspace{6\baselineskip}}
\AtBeginEnvironment{example}{\Needspace{10\baselineskip}}

\title[Medvedev Logic Is $\Pi^0_1$-Complete]{Medvedev Logic Is Not Decidable.\\ It Is $\Pi^0_1$-Complete. Who Would Have Guessed?}

\author{Pawel Pawlowski}
\address{CLPS, Ghent University, Belgium}

\date{Version 8, 9 September 2026}

\begin{document}
\begin{abstract}
This project began as an attempt to prove that Medvedev logic is decidable with the help of generative AI systems. The author (as well as the generative AI systems, or at least they claim to be since I have asked them) was surprised by its eventual conclusion.
We prove that Medvedev logic \(\ML\), the intermediate logic of finite
problems, is \(\Pi^0_1\)-complete under computable many-one reductions.
Consequently, \(\ML\) is not recursively enumerable, \emph{a fortiori}
undecidable, and admits no recursively enumerable sound and complete
proof calculus.

The proof connects the periodic domino problem with intuitionistic
formulas through a shared intermediate structure that we call a Wang--Medvedev pair.
Such a pair consists of a finite partially ordered set of roles together
with demands. Demands define the interaction between roles. A realization labels nonempty
subsets of a finite set with these roles, respecting the order and
satisfying the demands. We associate a pair with each finite Wang system
and show that it has a realization iff the system tiles a finite
torus. We then construct an intuitionistic formula that fails on some
finite Medvedev frame iff the same pair is realizable.
Realizability thus provides the link between periodic tilings and the
countermodels. 
\end{abstract}
\maketitle

\section{Introduction}
\label{sec:introduction}

Medvedev logic \(\ML\) was introduced by Medvedev in 1962 as
the logic of finite problems \cite{Medvedev1962}. We work with
its Kripke semantics. For a finite nonempty set \(K\), the
\emph{Medvedev frame} is \(F(K)=(\PowPlus{K},\supseteq)\),
where \(\PowPlus{K}\) consists of the nonempty subsets of \(K\).
The order is reverse inclusion, so moving upwards means taking a
smaller nonempty subset. The least world is \(K\), and the
maximal worlds are the singletons. A valuation is persistent when
truth at \(X\) is preserved at its nonempty subsets. With the
usual intuitionistic forcing clauses, we put
\begin{equation}
 \ML=\{\varphi:F_n\models\varphi\text{ for every }n\geq1\},
 \qquad F_n=F(\{0,\ldots,n-1\}).
\label{eq:old-5-1}
\end{equation}
Here frame validity means truth at every world under every persistent
valuation. Thus \(\ML\) is the intermediate logic of all finite
Medvedev frames.\footnote{Every finite nonempty carrier of size
\(n\) gives an isomorphic copy of \(F_n\). We recall the
forcing clauses in Section~\ref{sec:recognition}.}

The decision problem for \(\ML\) has a long pedigree.
Levin established its disjunction property \cite{Levin1969},
and Maksimova, Skvortsov, and Shehtman proved that it is not
finitely axiomatizable \cite{MSS1979}. However, whether it is
decidable, or admits a recursively enumerable axiomatization, remained
open \cite{ChenDing2025,Chen2026}.

A distinctly Polish thread runs through this history. The notion of
structural completeness goes back to Pogorzelski
\cite{Pogorzelski1971}. Prucnal proved that \(\ML\) is
structurally complete \cite{Prucnal1976} and subsequently settled
Friedman's Problem~41 by proving the existence of a greatest
structurally complete intermediate logic with the disjunction
property \cite{Prucnal1979}. Szatkowski identified fragments on
which \(\ML\) and intuitionistic logic coincide, including the
disjunction-free fragment \cite{Szatkowski1981}.
\L{}azarz later characterized \(\ML\) using Kubi\'nski's
frames \cite{Lazarz2013}.\footnote{For a modern presentation of
the Levin and Prucnal theorems, see \cite{Prenosil2024}. Recent
work also includes the logics of individual Medvedev frames
\cite{ChenDing2025} and a finitely axiomatized logic contained
in \(\ML\) with the disjunction property \cite{Chen2026}.}

In this paper we prove that \(\ML\) is \(\Pi^0_1\)-complete
under computable many-one reductions (Theorem~\ref{thm:main}).
Consequently, it is not recursively enumerable, \emph{a fortiori}
undecidable, and admits no recursively enumerable sound and complete
proof calculus. The upper bound follows by enumerating finite
countermodels. The lower bound turns on a translation from finite
Wang systems to formulas, which we describe below. There is also a
connection with the logic of questions: \(\ML\) coincides with
the schematic fragment of inquisitive logic
\cite{CiardelliRoelofsen2011Inquisitive}, consisting of the
formulas whose substitution instances are all inquisitively valid.

We fix a finite Wang system \(D\). It specifies tile types and
the allowed horizontal and vertical neighbours. The question is
whether these tiles cover some finite rectangular torus, with the
dimensions and the assignment of tiles left to be found. From
\(D\), we define a finite poset of roles and a list of demands,
\(\WM(D)=(P_D,\D_D)\). A realization assigns roles to some
nonempty subsets of a finite set \(K\), subject to containment
restrictions. Each demand requires a labelled subset of the union
of two input sets, called a \emph{response}. Demands are unordered
pairs of roles, so their requirements are unchanged when the inputs
are exchanged.

Suppose first that a realization is given. Responses generate a
column cycle and a row cycle, whose lengths give the dimensions
of the torus. A response to each column--row pair selects the tile
at that position. At this stage we still need to prove that the
matching rules hold. We take adjacent cells and suppose their
tile types form a forbidden pair. Further responses give sets
with the corresponding Expected and Candidate roles. Because the
tile pair is forbidden, the code contains a demand between these
roles. Their required response would violate the containment
restrictions, giving a contradiction. This proves compatibility
in both directions (Section~\ref{sec:realization-to-torus}).

For the converse, we fix a torus tiling and construct the subsets
that realize the code. Each set is described by the symbols it
omits from copies of a common finite set. Shared omissions record
the column and row comparisons. The construction must supply a
response to every demanded pair, including pairs associated with
unrelated positions. Section~\ref{sec:torus-to-realization} gives
the lists of omissions and checks these requirements.

It remains to express realizability by a formula. We introduce one
variable for each role, intended to say that no subset of the
current world has that role. Using these variables, we construct
\(\alpha_D\) so that a realization yields a refutation on a
finite Medvedev frame, and any such refutation yields a realization.
Hence \(D\) tiles a finite torus iff
\(\alpha_D\notin\ML\). The effective periodic domino theorem of
Gurevich and Koryakov \cite{GurevichKoryakov1972}, in the form given
by Jeandel \cite{Jeandel2010}, supplies a Wang system that tiles a
finite torus iff a given machine halts. Applying our
translation to that system gives a formula in \(\ML\) exactly
when the machine does not halt, as required.

The paper is organized as follows. Section~\ref{sec:wmpair} fixes the
code \(\WM(D)\) and the notion of realization. Its main product is the
response-role table, which settles once and for all which roles a
response can have. Sections~\ref{sec:realization-to-torus} and~\ref{sec:torus-to-realization} prove the two directions of the
equivalence between realizability and torus tilability. The first is
a short extraction argument. The second, the construction, is where
most of the bookkeeping lives. Section~\ref{sec:recognition} homes in
on the link between realizability and formulas of Medvedev logic: it
translates realizability into a formula and proves the recognition
theorem on a fixed carrier. Section~\ref{sec:assembly} then assembles
the complexity result in a few lines. Proofs of the opening lemmas
and of the response-role table are deferred to
Appendix~\ref{app:opening-proofs}.

\section{The Wang system and its Wang--Medvedev code}
\label{sec:wmpair}

We begin with the notation used throughout the paper. We write
\(\mathbb N=\{0,1,\ldots\}\), \(\PowPlus{K}\) for the
nonempty subsets of \(K\), and \(\dot\cup\) for disjoint
union. For a map \(q:L\to M\), its direct and inverse images
are \(q[X]=\{q(x):x\in X\}\) and
\(q^{-1}[Y]=\{x\in L:q(x)\in Y\}\), for
\(X\subseteq L\) and \(Y\subseteq M\).

A \emph{poset} is a set with a reflexive, antisymmetric, transitive
relation \(\leq\). The \emph{principal upset} of \(a\) is
\(\up a=\{b:a\leq b\}\). An \emph{upset} is closed
upwards under \(\leq\). A least point lies below every point,
and a maximal point has no strictly greater point.

Our propositional language has countably infinitely many variables and the
connectives \(\bot,\wedge,\vee,\to\). We write
\(\neg\varphi\) for \(\varphi\to\bot\) and \(\top\)
for \(\bot\to\bot\). Finite conjunctions and disjunctions
are iterated binary ones, with empty cases \(\top\) and
\(\bot\), respectively.

For \(m\geq1\), we use
\(\mathbb Z_m=\{0,\ldots,m-1\}\) with addition modulo
\(m\). So the successor of \(m-1\) is \(0\), and period
one is allowed. Superscripts \(h,v\) and subscripts
\(\mathrm{col},\mathrm{row}\) indicate the horizontal and vertical
directions. Notation depending on the code, such as
\(\role(A)\), \(\Resp(A,B)\), and \(A[\chi]\), is
introduced with the realization conditions below.

Our encoding assigns abstract \emph{roles} to some worlds of a
finite Medvedev frame \(F(K)\), hence to nonempty subsets of \(K\).
Intuitively, the role says what a set is used for, such as representing a column or
a cell of a particular tile type. Different subsets may have the same
role. The role poset specifies which containments are permitted, and
its demands specify which response sets must exist.

\begin{definition}[Finite Wang system]\label{def:wang}
A \emph{finite Wang system} \cite{Wang1961} is a triple \(D=(T,H,V)\), where \(T\) is a
finite nonempty set of tile types, \(H\subseteq T\times T\) is the
horizontal compatibility relation, and \(V\subseteq T\times T\) is the
vertical compatibility relation.  Thus \((t,u)\in H\) means that \(u\)
may occur immediately to the right of \(t\), while \((t,u)\in V\) means
that \(u\) may occur immediately above \(t\).  Both relations are
directed, so the order of \(t,u\) matters.
\end{definition}

\begin{definition}[Finite torus tiling]\label{def:torus}
A \emph{finite torus tiling} of \(D=(T,H,V)\) is a map
\(\tau:\mathbb Z_m\times\mathbb Z_n\to T\), for some \(m,n\ge1\), such
that \((\tau(i,j),\tau(i+1,j))\in H\) and
\((\tau(i,j),\tau(i,j+1))\in V\) for every \((i,j)\), with both
coordinates read cyclically.  Thus the last column is adjacent to the
first, and the last row is adjacent to the first.
\end{definition}

The system \(D\) specifies the available tiles and the matching rules.
A torus tiling chooses dimensions and an assignment satisfying those rules.
We fix an arbitrary \(D\) and define its code below. The code depends
only on \(D\), so it is fixed before any tiling or realization is chosen.


\begin{definition}[Wang--Medvedev pair]\label{def:wmpair}
Fix a finite Wang system \(D=(T,H,V)\). The \emph{Wang--Medvedev pair}
associated with \(D\), also called the \emph{code} of \(D\), is
\(\WM(D)=(P_D,\D_D)\), defined as follows.

\begin{enumerate}[label=\textbf{(\arabic*)},leftmargin=2.8em,nosep]

\item \textbf{Points.}
The poset \(P_D=\{r\}\,\dot\cup\,\Mid_D\,\dot\cup\,\MaxPts_D\)
has three levels: the root \(r\), the Mid-points, and the Max-points.
Put
\(\mathcal G=\{\StepL{h},\StepL{v}\}\cup
\{\SrcL{h}{t},\allowbreak\NbrL{h}{t},\allowbreak
\SrcL{v}{t},\allowbreak\NbrL{v}{t}:t\in T\}\).

The set \(\MaxPts_D\) consists of the four roles
\(\sa{h},\nx{h},\sa{v},\nx{v}\), the tags \(\tg{h},\tg{v}\), and a
private key \(\keyc{G}\) for every \(G\in\mathcal G\). Thus
\(|\MaxPts_D|=4|T|+8\).

The set \(\Mid_D\) consists of the horizontal generator roles
\(\CPos,\CSep\) and mismatch roles \(\NFL{h},\SFL{h}\), their vertical
analogues \(\RPos,\RSep\) and \(\NFL{v},\SFL{v}\), and the step roles
\(\StepL{h},\StepL{v}\). For every \(t\in T\), it also contains the
cell role \(\CellL{t}\), the horizontal projections
\(\SrcL{h}{t},\NbrL{h}{t}\), the vertical projections
\(\SrcL{v}{t},\NbrL{v}{t}\), and the expected-neighbour roles
\(\ExpL{h}{t},\ExpL{v}{t}\). Finally, for every \(G\in\mathcal G\),
it contains a provider role \(\KeyL{G}\). All names introduced above
denote pairwise distinct points, and \(|\Mid_D|=11|T|+12\).

\item \textbf{Order.}
The root \(r\) is the least point. Distinct Mid-points are incomparable,
as are distinct Max-points. The remaining strict comparabilities are
exactly those given by the following table: the right-hand entry lists
all Max-points lying above the Mid-point or Mid-points on the left.

\begin{center}
\small
\renewcommand{\arraystretch}{1.08}
\begin{tabularx}{.97\linewidth}
{@{}>{\raggedright\arraybackslash}p{.30\linewidth}
>{\raggedright\arraybackslash}X@{}}
\toprule
\textbf{Mid-role(s)} & \textbf{Max-points above}\\
\midrule
\(\CPos,\CSep\) &
  \(\{\sa{h},\nx{h},\tg{h}\}\)\\
\(\NFL{h}\) &
  \(\{\nx{h}\}\)\\
\(\SFL{h}\) &
  \(\{\sa{h}\}\)\\
\(\RPos,\RSep\) &
  \(\{\sa{v},\nx{v},\tg{v}\}\)\\
\(\NFL{v}\) &
  \(\{\nx{v}\}\)\\
\(\SFL{v}\) &
  \(\{\sa{v}\}\)\\
\(\CellL{t}\) &
  \(\{\sa{h},\nx{h},\sa{v},\nx{v}\}\)\\
\(\SrcL{h}{t}\) &
  \(\{\sa{h},\sa{v},\keyc{\SrcL{h}{t}}\}\)\\
\(\NbrL{h}{t}\) &
  \(\{\nx{h},\sa{v},\keyc{\NbrL{h}{t}}\}\)\\
\(\SrcL{v}{t}\) &
  \(\{\sa{v},\sa{h},\keyc{\SrcL{v}{t}}\}\)\\
\(\NbrL{v}{t}\) &
  \(\{\nx{v},\sa{h},\keyc{\NbrL{v}{t}}\}\)\\
\(\StepL{h}\) &
  \(\{\sa{h},\nx{h},\keyc{\StepL{h}}\}\)\\
\(\StepL{v}\) &
  \(\{\sa{v},\nx{v},\keyc{\StepL{v}}\}\)\\
\(\ExpL{h}{t}\) &
  \(\{\nx{h},\sa{v},\keyc{\SrcL{h}{t}},
     \allowbreak\keyc{\StepL{h}}\}\)\\
\(\ExpL{v}{t}\) &
  \(\{\nx{v},\sa{h},\keyc{\SrcL{v}{t}},
     \allowbreak\keyc{\StepL{v}}\}\)\\
\(\KeyL{G}\) &
  \(\{\keyc{G}\}\quad(G\in\mathcal G)\)\\
\bottomrule
\end{tabularx}
\end{center}

Here \(t\in T\). Every row is nonempty and every Max-point occurs in at
least one row. For \(\lambda\in\Mid_D\), write
\(\MaxOf{\lambda}=\{\chi\in\MaxPts_D:\lambda<\chi\}\); thus
\(\MaxOf{\lambda}\) is exactly the corresponding entry of the table.
There are no other strict comparabilities. We call
\(\MaxOf{\lambda}\) the \emph{support} of the Mid-role \(\lambda\).
Two different roles may have the same support. For example,
\(\CPos\) and \(\CSep\) are distinct incomparable points, even though
they have the same three Max-points above them.

\item \textbf{Demands.}
The demand set \(\D_D\) consists of the following unordered pairs of
distinct Mid-points. Thus \(\{\lambda,\gamma\}=\{\gamma,\lambda\}\),
and exchanging the two inputs does not give a new demand.
\begin{enumerate}[label=\textbf{(D\arabic*)},leftmargin=3.2em,nosep]
\raggedright

\item \(\{\CPos,\NFL{h}\}\), \(\{\CSep,\SFL{h}\}\),
\(\{\RPos,\NFL{v}\}\), and \(\{\RSep,\SFL{v}\}\).

\item \(\{\CPos,\RPos\}\).

\item For each \(t\in T\),
\(\{\CellL{t},\KeyL{\SrcL{h}{t}}\}\),
\(\{\CellL{t},\KeyL{\NbrL{h}{t}}\}\),
\(\{\CellL{t},\KeyL{\SrcL{v}{t}}\}\), and
\(\{\CellL{t},\KeyL{\NbrL{v}{t}}\}\).

\item \(\{\CSep,\KeyL{\StepL{h}}\}\) and
\(\{\RSep,\KeyL{\StepL{v}}\}\).

\item For each \(t\in T\), \(\{\SrcL{h}{t},\StepL{h}\}\) and
\(\{\SrcL{v}{t},\StepL{v}\}\).

\item For each \((t,u)\notin H\),
\(\{\ExpL{h}{t},\NbrL{h}{u}\}\).

\item For each \((t,u)\notin V\),
\(\{\ExpL{v}{t},\NbrL{v}{u}\}\).

\end{enumerate}

The pairs in (D6) and (D7) are unordered as well. The tile order is
recorded by the different Expected and Candidate roles. For example,
exchanging the two entries of
\(\{\ExpL{h}{t},\NbrL{h}{u}\}\) preserves that demand, whereas
exchanging \(t,u\) gives
\(\{\ExpL{h}{u},\NbrL{h}{t}\}\), a different pair when \(t\ne u\).

The seven families are disjoint, and distinct tile indices give distinct
pairs. Hence
\(|\D_D|=6|T|+7+|T^2\setminus H|+|T^2\setminus V|\).

\end{enumerate}
\end{definition}

From now on, we abbreviate \(P=P_D\), \(\Mid=\Mid_D\),
\(\MaxPts=\MaxPts_D\), and \(\D=\D_D\).
In particular, \(\up\lambda=\{\lambda\}\cup\MaxOf{\lambda}\)
for every Mid-role \(\lambda\).

We read the roles in the order in which the extraction proof will use
 them. First come the column and row generators, then the cells, and
 finally the comparison roles.

In the following support diagrams, blue boxes are Mid-roles and orange boxes are
Max-roles, including the private keys. Every box names an abstract
point of \(P_D\). A solid line from a lower box \(\lambda\) to an
upper box \(\chi\) means \(\lambda<\chi\) in the role poset.
The root lies below every displayed role and is omitted. The diagrams
show order relations; the response rules will concern nonempty subsets
of a carrier \(K\) that bear these roles.

\definecolor{wmMid}{HTML}{2166AC}
\definecolor{wmMax}{HTML}{B45F06}
\tikzset{
  wm mid/.style={draw=wmMid,fill=wmMid!9,rounded corners=2pt,
    font=\scriptsize,align=center,inner sep=3pt,outer sep=1pt},
  wm max/.style={draw=wmMax,fill=wmMax!10,rounded corners=2pt,
    font=\scriptsize,align=center,inner sep=3pt,outer sep=1pt},
  wm order/.style={draw=black,line width=.4pt},
  wm panel/.style={font=\small,align=center}
}

Family (D1) generates an alternating cycle of column and separator
sets, and an analogous cycle of row and separator sets. The respective
numbers of column and row positions will be \(m\) and \(n\). These two
cycles supply the torus indices \(\mathbb Z_m\) and \(\mathbb Z_n\).
Figure~\ref{fig:wm-hasse} shows the roles used for this first task.

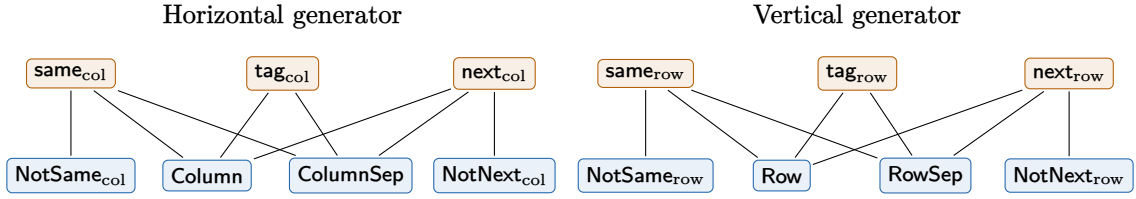
\begin{figure}[!ht]
\centering
\begin{tikzpicture}
  \begin{scope}
    \node[wm panel] at (2.8,2.1) {Horizontal generator};
    \node[wm max] (hs) at (0,1.35) {$\sa{h}$};
    \node[wm max] (ht) at (2.8,1.35) {$\tg{h}$};
    \node[wm max] (hn) at (5.6,1.35) {$\nx{h}$};
    \node[wm mid] (hns) at (0,0) {$\SFL{h}$};
    \node[wm mid] (hc) at (1.8,0) {$\CPos$};
    \node[wm mid] (hq) at (3.7,0) {$\CSep$};
    \node[wm mid] (hnn) at (5.6,0) {$\NFL{h}$};
    \foreach \a in {hc,hq} {
      \foreach \b in {hs,ht,hn} {\draw[wm order] (\a)--(\b);}}
    \draw[wm order] (hns)--(hs);
    \draw[wm order] (hnn)--(hn);
  \end{scope}
  \begin{scope}[xshift=7.6cm]
    \node[wm panel] at (2.8,2.1) {Vertical generator};
    \node[wm max] (vs) at (0,1.35) {$\sa{v}$};
    \node[wm max] (vt) at (2.8,1.35) {$\tg{v}$};
    \node[wm max] (vn) at (5.6,1.35) {$\nx{v}$};
    \node[wm mid] (vns) at (0,0) {$\SFL{v}$};
    \node[wm mid] (vr) at (1.8,0) {$\RPos$};
    \node[wm mid] (vw) at (3.7,0) {$\RSep$};
    \node[wm mid] (vnn) at (5.6,0) {$\NFL{v}$};
    \foreach \a in {vr,vw} {
      \foreach \b in {vs,vt,vn} {\draw[wm order] (\a)--(\b);}}
    \draw[wm order] (vns)--(vs);
    \draw[wm order] (vnn)--(vn);
  \end{scope}
\end{tikzpicture}
\caption{The roles used to generate the two cyclic coordinate systems.
Their supports allow the demands in (D1) to produce the required
alternating responses.}
\label{fig:wm-hasse}
\end{figure}

Once these positions have been chosen, family (D2) requests a cell
response to each column--row pair. The response has role
\(\CellL{t}\) for some \(t\in T\), and we assign tile type \(t\)
to that position. Figure~\ref{fig:wm-cell-support} shows the common
support of these cell roles. The tile label \(t\) belongs to the
Mid-role. The four Max-roles above it will be used to compare
neighbouring positions.

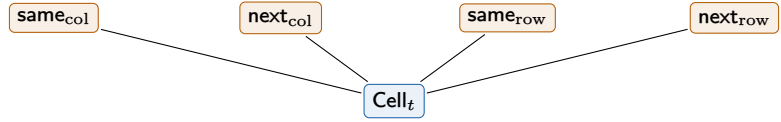
\begin{figure}[!ht]
\centering
\begin{tikzpicture}
  \node[wm max] (sc) at (0,1.1) {$\sa{h}$};
  \node[wm max] (nc) at (3,1.1) {$\nx{h}$};
  \node[wm max] (sr) at (6,1.1) {$\sa{v}$};
  \node[wm max] (nr) at (9,1.1) {$\nx{v}$};
  \node[wm mid] (cell) at (4.5,0) {$\CellL{t}$};
  \foreach \a in {sc,nc,sr,nr} {\draw[wm order] (cell)--(\a);}
\end{tikzpicture}
\caption{The cell role assigns a tile type to a chosen column--row
position. There is one distinct role \(\CellL{t}\) for each \(t\in T\).}
\label{fig:wm-cell-support}
\end{figure}

Repeated labels name the same point of \(P_D\).\footnote{For
example, the lines from \(\CPos\) and \(\CellL{t}\) to
\(\sa{h}\) end at the same Max-point. Drawing it in separate
panels does not create new points. The same convention applies to
the private keys.}

\Needspace{5\baselineskip}
The remaining roles form the comparison machinery. Choosing cell
responses gives a tile assignment, but nothing yet guarantees that
neighbouring tile types satisfy \(H\) and \(V\), and checking this
is the \emph{raison d'\^etre} of the roles that follow. For a horizontal
edge, the Source retains the left tile type, the ColStep supplies
information from the intervening column separator, and the Expected
response combines the two. The Candidate retains the tile type chosen
in the right cell. We then compare Expected and Candidate. If their
tile pair were forbidden, (D6) would require a response that the
coordinate containments make impossible.

Figure~\ref{fig:wm-comparison-support} separates the supports used in
this comparison. A provider \(\KeyL{G}\) is a Mid-role whose only
Max-role is \(\keyc{G}\). A set bearing this provider role supplies
the private-key points needed to obtain the auxiliary role \(G\)
from a chosen cell or separator. Families (D3)--(D4) make these
requests. We then use (D5) to form the Expected response. The detailed
proof below explains both the response roles and the required
containments.

\begin{figure}[!ht]
\centering
\begin{tikzpicture}
  \begin{scope}
    \node[wm panel] at (2.65,1.95) {Source and its provider};
    \node[wm max] (ss) at (0,1.2) {$\sa{h}$};
    \node[wm max] (sr) at (2.5,1.2) {$\sa{v}$};
    \node[wm max] (sk) at (5.3,1.2) {$\keyc{\SrcL{h}{t}}$};
    \node[wm mid] (src) at (1.25,0) {$\SrcL{h}{t}$};
    \node[wm mid] (sprov) at (5.3,0) {$\KeyL{\SrcL{h}{t}}$};
    \foreach \a in {ss,sr,sk} {\draw[wm order] (src)--(\a);}
    \draw[wm order] (sprov)--(sk);
  \end{scope}
  \begin{scope}[xshift=7.6cm]
    \node[wm panel] at (2.65,1.95) {Step and its provider};
    \node[wm max] (ts) at (0,1.2) {$\sa{h}$};
    \node[wm max] (tn) at (2.5,1.2) {$\nx{h}$};
    \node[wm max] (tk) at (5.3,1.2) {$\keyc{\StepL{h}}$};
    \node[wm mid] (step) at (1.25,0) {$\StepL{h}$};
    \node[wm mid] (tprov) at (5.3,0) {$\KeyL{\StepL{h}}$};
    \foreach \a in {ts,tn,tk} {\draw[wm order] (step)--(\a);}
    \draw[wm order] (tprov)--(tk);
  \end{scope}
  \begin{scope}[yshift=-3.2cm]
    \node[wm panel] at (2.65,1.95) {Expected neighbour of the source tile};
    \node[wm max] (en) at (-.2,1.2) {$\nx{h}$};
    \node[wm max] (er) at (1.7,1.2) {$\sa{v}$};
    \node[wm max] (esk) at (3.8,1.2) {$\keyc{\SrcL{h}{t}}$};
    \node[wm max] (etk) at (6,1.2) {$\keyc{\StepL{h}}$};
    \node[wm mid] (exp) at (2.9,0) {$\ExpL{h}{t}$};
    \foreach \a in {en,er,esk,etk} {\draw[wm order] (exp)--(\a);}
  \end{scope}
  \begin{scope}[xshift=7.6cm,yshift=-3.2cm]
    \node[wm panel] at (2.65,1.95) {Candidate and its provider};
    \node[wm max] (cn) at (0,1.2) {$\nx{h}$};
    \node[wm max] (cr) at (2.5,1.2) {$\sa{v}$};
    \node[wm max] (ck) at (5.3,1.2) {$\keyc{\NbrL{h}{u}}$};
    \node[wm mid] (cand) at (1.25,0) {$\NbrL{h}{u}$};
    \node[wm mid] (cprov) at (5.3,0) {$\KeyL{\NbrL{h}{u}}$};
    \foreach \a in {cn,cr,ck} {\draw[wm order] (cand)--(\a);}
    \draw[wm order] (cprov)--(ck);
  \end{scope}
\end{tikzpicture}
\caption{Horizontal comparison roles for tile types \(t,u\in T\).
Providers remain Mid-roles; their private keys are Max-roles.
For vertical comparison, exchange \(h\) with \(v\), column with row,
and ColStep with RowStep. Family (D7) then tests forbidden vertical
pairs.}
\label{fig:wm-comparison-support}
\end{figure}
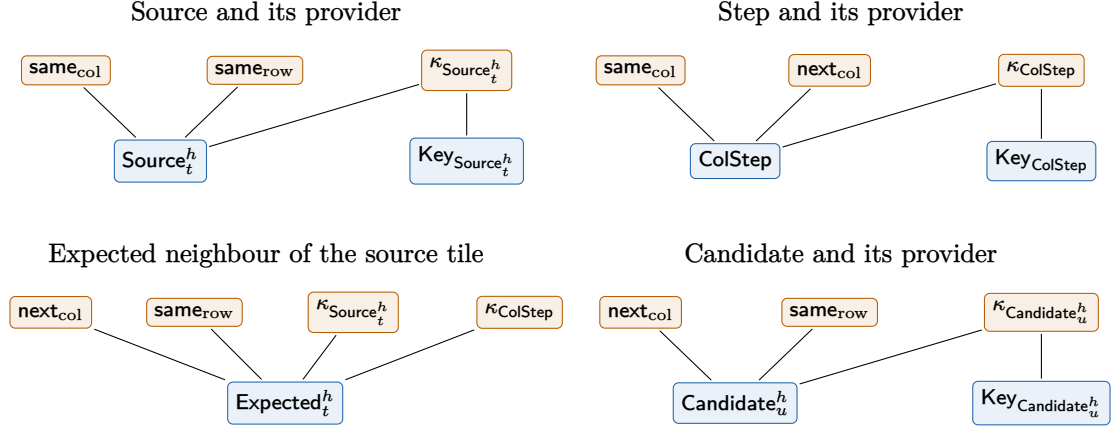

\FloatBarrier
\begin{definition}[Finite realization]\label{def:realization}
A \emph{finite realization} of \(\WM(D)\) consists of a finite
nonempty carrier \(K\) and a partial function
\(\role:\PowPlus{K}\rightharpoonup P\) satisfying the following
conditions. A subset is \emph{role-bearing} when this function is
defined on it.
\begin{description}[leftmargin=3.4em,labelwidth=2.7em,labelsep=.5em,
                    style=multiline,nosep]
\item[(H0)] \textbf{Base and singleton roles.}
The whole carrier has role \(r\), and every singleton has a Max-role:
\(\role(K)=r\) and \(\role(\{k\})\in\MaxPts\) for every \(k\in K\).

\item[(H1)] \textbf{Containment and order.}
(a) For every role-bearing \(A\) with \(\role(A)=\lambda\), and
every \(\gamma\in P\) with \(\lambda\leq\gamma\), there is a
nonempty \(C\subseteq A\) with \(\role(C)=\gamma\).
(b) For every two role-bearing sets \(A,B\), if
\(\role(A)=\lambda\), \(\role(B)=\gamma\), and \(B\subseteq A\),
then \(\lambda\leq\gamma\).

\item[(H2)] \textbf{Demand responses.}
For every demand \(\{\lambda,\gamma\}\in\D\), and every two
sets \(X,Y\) with \(\role(X)=\lambda\) and \(\role(Y)=\gamma\),
there is a nonempty role-bearing \(Z\subseteq X\cup Y\) whose role
\(\delta=\role(Z)\) lies outside \(\up\lambda\cup\up\gamma\).
Such a set \(Z\) is a \emph{response} to \(X,Y\); its role is a
\emph{legal response role} for this demand.
\end{description}
\end{definition}

In a fixed realization, consider role-bearing sets \(A,B\) whose
roles form a demand. We choose one response for each unordered
input pair and denote it by \(\Resp(A,B)=\Resp(B,A)\). By (H2),
\(\varnothing\ne\Resp(A,B)\subseteq A\cup B\), and
\(\role(\Resp(A,B))\notin\up\role(A)\cup\up\role(B)\).
The notation is used only for demanded pairs and refers to these
fixed choices.

For any subset \(A\subseteq K\) and \(\chi\in\MaxPts\), we write
\(A[\chi]=\{k\in A:\role(\{k\})=\chi\}\).
Thus \(A[\sa{h}]\) is the set of elements of \(A\) whose
singletons have role \(\sa{h}\). The subset \(A[\chi]\) may be
empty or have no assigned role. The following fact explains how
these parts behave inside an input union.

\begin{fact}[Parts of a union]\label{fact:parts-union}
Let \(X,Y,Z\subseteq K\) and \(\chi\in\MaxPts\).
If \(Z\subseteq X\cup Y\), then
\(Z[\chi]\subseteq X[\chi]\cup Y[\chi]\).
In particular, if \(Y[\chi]=\varnothing\), then
\(Z[\chi]\subseteq X[\chi]\).
If also \(Z=Z[\chi]\), this yields \(Z\subseteq X\).
\end{fact}
\begin{proof}
Every element of \(Z[\chi]\) belongs to \(X\) or \(Y\) and its
singleton has role \(\chi\), so it belongs to \(X[\chi]\) or
\(Y[\chi]\). The two consequences follow immediately.
\end{proof}

By (H0), every point of \(K\) has a singleton role. Moreover
\(|K|\geq2\), since otherwise one set would carry both the root role
and a Max-role. Clause (H1)(a) supplies representatives of larger
roles inside an assigned set, while (H1)(b) restricts which assigned
sets may contain one another. In particular, sets with distinct
Mid-roles never contain one another.

Notice that condition (H2) applies to \emph{every} pair of sets
carrying a demanded pair of roles. The response may depend on the
chosen sets, and it need only lie in their union. This is all that
Section~\ref{sec:torus-to-realization} will have to supply.
A Mid-role is legal iff it differs from both input roles.

\begin{lemma}[Properties of realizations]
\label{lem:basic-concrete}
The following statements hold in every realization.
\begin{enumerate}[label=\textup{(\alph*)},leftmargin=2.3em,nosep]
\item If \(\role(C)=\chi\) for \(\chi\in\MaxPts\), then every
      \(k\in C\) satisfies \(\role(\{k\})=\chi\).
\item If \(\role(A)=\lambda\) for \(\lambda\in\Mid\), then every
      \(k\in A\) satisfies \(\role(\{k\})\in\MaxOf{\lambda}\).
      Conversely, for every \(\chi\in\MaxOf{\lambda}\), the set
      \(A\) contains at least one element \(k\) with
      \(\role(\{k\})=\chi\).  Hence the singleton roles occurring in
      \(A\) are exactly the Max-points above \(\lambda\).
\item If \(\role(A)=\lambda\), \(\role(B)=\gamma\), and
      \(\lambda\nleq\gamma\), then \(B\nsubseteq A\).  In particular,
      a set with one Mid-role cannot be contained in a set with a
      distinct Mid-role.
\end{enumerate}
\end{lemma}

The proof is given in \hyperref[proof:basic-concrete]{Appendix~\ref*{app:opening-proofs}}.

\begin{corollary}[Every role occurs; one role per labelled set]
\label{cor:roles}
Every \(\lambda\in P\) has a nonempty representative
\(A\subseteq K\) with \(\role(A)=\lambda\). Every role-bearing
subset of \(K\) has exactly one role.
\end{corollary}

\begin{proof}
By (H0), \(\role(K)=r\). Since \(r\leq\lambda\), condition
(H1)(a) supplies the required nonempty subset of \(K\).
Uniqueness of a set's role follows because \(\role\) is a partial
function.
\end{proof}

A role may have several representatives, and some nonempty subsets
may be unlabelled. Corollary~\ref{cor:roles} guarantees representatives
somewhere in \(K\). To serve as a response to a particular pair,
a representative must also satisfy the containment and role
requirements in (H2).

For instance, if $\role(A)=\CellL{t}$, then the singleton roles
occurring in $A$ are exactly $\sa{h}$, $\nx{h}$, $\sa{v}$, and
$\nx{v}$, and each of these four roles occurs at least once.

These facts let us check response roles using the support table.
An input union contains only singleton roles supplied by its two inputs.
Moreover, a legal response cannot obtain all its elements from just
one input, by Lemma~\ref{lem:basic-concrete}(c).

\begin{lemma}[Response test]\label{lem:response-test}
Let \(\{\lambda,\gamma\}\in\D_D\), with \(\role(X)=\lambda\) and
\(\role(Y)=\gamma\), and let \(Z\) be a response to \(X\) and \(Y\),
with \(\role(Z)=\delta\). Then:
\begin{enumerate}[label=\textup{(\alph*)},nosep]
\item \(\delta\in\Mid_D\);
\item \(\MaxOf{\delta}\subseteq\MaxOf{\lambda}\cup\MaxOf{\gamma}\);
\item \(\MaxOf{\delta}\cap\MaxOf{\lambda}\neq\varnothing\) and
      \(\MaxOf{\delta}\cap\MaxOf{\gamma}\neq\varnothing\).
\end{enumerate}
\end{lemma}

The proof is given in \hyperref[proof:response-test]{Appendix~\ref*{app:opening-proofs}}.

Call a Mid-role \(\delta\) \emph{admissible} for a pair
\(\{\lambda,\gamma\}\) if it differs from both input roles,
its support is contained in the union of their supports, and its
support meets each input support. These are three checks in the finite
support table. The response test says that every actual response has
an admissible role. Admissibility is a necessary condition on a
response, not a sufficient one. It does not \emph{per se} supply a
concrete response set.

We can exchange the two axes in any such calculation.\footnote{Formally,
let \(\sigma\) exchange \(\CPos,\CSep\) with \(\RPos,\RSep\)
and exchange the superscripts \(h,v\) and the subscripts
\(\mathrm{col},\mathrm{row}\) in all other role names, tags and
keys (so \(\StepL{h}\) with \(\StepL{v}\)), while fixing \(r\)
and the cell roles. The support table gives
\(\MaxOf{\sigma\lambda}=\sigma[\MaxOf{\lambda}]\), where
\(\sigma[A]=\{\sigma(a):a\in A\}\). Thus \(\sigma\) is an
order automorphism of \(P_D\), and it preserves all three
admissibility checks. This remains true when \(H\ne V\):
admissibility concerns the supports of a pair, irrespective of whether
that pair is a demand.}

For a tile \(t\), write
\(G_t\in\{\SrcL{h}{t},\NbrL{h}{t},\SrcL{v}{t},\NbrL{v}{t}\}\).
The screening procedure gives the following table.

\begin{proposition}[Forced response-role table]\label{prop:forced-table}
For a demanded input pair, every response has one of the roles in the
corresponding row below. The displayed roles are exactly the admissible
ones. Their admissibility does not guarantee that each has a witness
for every concrete input pair.
\end{proposition}

\begin{center}
\small
\renewcommand{\arraystretch}{1.18}
\begin{tabularx}{\textwidth}{@{}c >{\raggedright\arraybackslash}X >{\raggedright\arraybackslash}p{.34\textwidth}@{}}
\toprule
& \textbf{demand} & \textbf{admissible response roles}\\
\midrule
(M1) & \(\{\CPos,\NFL{h}\}\) & \(\{\CSep\}\)\\
(M2) & \(\{\CSep,\SFL{h}\}\) & \(\{\CPos\}\)\\
(M3) & \(\{\RPos,\NFL{v}\}\) & \(\{\RSep\}\)\\
(M4) & \(\{\RSep,\SFL{v}\}\) & \(\{\RPos\}\)\\
(M5) & \(\{\CPos,\RPos\}\) & \(\{\CellL{t}:t\in T\}\)\\
(M6) & \(\{\CellL{t},\KeyL{G_t}\}\) & \(\{G_t\}\)\\
(M7) & \(\{\CSep,\KeyL{\StepL{h}}\}\) & \(\{\StepL{h}\}\)\\
(M8) & \(\{\RSep,\KeyL{\StepL{v}}\}\) & \(\{\StepL{v}\}\)\\
(M9) & \(\{\SrcL{h}{t},\StepL{h}\}\) & \(\{\SFL{h},\ExpL{h}{t}\}\)\\
(M10) & \(\{\SrcL{v}{t},\StepL{v}\}\) & \(\{\SFL{v},\ExpL{v}{t}\}\)\\
(M11) & \(\{\ExpL{h}{t},\NbrL{h}{u}\}\), \((t,u)\notin H\)
      & \(\{\NFL{h},\SFL{v}\}\)\\
(M12) & \(\{\ExpL{v}{t},\NbrL{v}{u}\}\), \((t,u)\notin V\)
      & \(\{\NFL{v},\SFL{h}\}\)\\
\bottomrule
\end{tabularx}
\end{center}

The proof is given in \hyperref[proof:forced-table]{Appendix~\ref*{app:opening-proofs}}.

The table separates two tasks. It determines the possible roles from
the abstract code. The following section uses concrete containments
to rule out mismatch responses when the coordinates really are
neighbours.

\section{From a finite realization to a torus tiling}
\label{sec:realization-to-torus}

We fix a Wang system \(D=(T,H,V)\) and a finite realization
\((K,\role)\) of \(\WM(D)\). We will obtain a finite torus tiling of \(D\):
responses first give column and row cycles, then assign tiles to
their intersections. We prove that these tiles satisfy \(H,V\),
including across the cyclic boundaries.

We begin with the columns. We fix nonempty auxiliary sets
\(R^h,S^h\subseteq K\) such that
\(\role(R^h)=\NFL{h}\) and
\(\role(S^h)=\SFL{h}\).\footnote{These sets exist by
Corollary~\ref{cor:roles}. The support table and
Lemma~\ref{lem:basic-concrete}(b) give
\(R^h=R^h[\nx{h}]\) and \(S^h=S^h[\sa{h}]\).
The vertical helpers below exist for the same reason.}
Take a set \(A\) with \(\role(A)=\CPos\), again by
Corollary~\ref{cor:roles}, and put \(B=\Resp(A,R^h)\), so that
\(\role(B)=\CSep\) by (M1). Then put \(A'=\Resp(B,S^h)\), which
has role \(\CPos\) by (M2). These two responses thus lead from one
column set, through a separator, to another.

Since \(\Resp\) is fixed, this procedure defines a successor
function on the finite family of sets with role \(\CPos\) or
\(\CSep\). Iteration eventually repeats a set. We discard the
initial path and retain a cycle
\(P_0,Q_0,\ldots,P_{m-1},Q_{m-1}\), starting at a column set.
Here \(m\geq1\), \(\role(P_i)=\CPos\), and
\(\role(Q_i)=\CSep\). For every \(i\in\mathbb Z_m\),
\(Q_i=\Resp(P_i,R^h)\) and
\(P_{i+1}=\Resp(Q_i,S^h)\).
We use \(P_i\) to represent column \(i\), with \(Q_i\) the
separator leading to column \(i+1\).

Rows are handled in the same way. Fix nonempty
\(R^v,S^v\subseteq K\) with
\(\role(R^v)=\NFL{v}\) and \(\role(S^v)=\SFL{v}\).
Using (M3)--(M4), we obtain a cycle
\(V_0,W_0,\ldots,V_{n-1},W_{n-1}\), where \(n\geq1\),
\(\role(V_j)=\RPos\), and \(\role(W_j)=\RSep\).
Thus \(W_j=\Resp(V_j,R^v)\) and
\(V_{j+1}=\Resp(W_j,S^v)\). We use \(V_j\) for row \(j\),
with \(W_j\) the separator leading to row \(j+1\).\footnote{Here
\(V_j\subseteq K\), whereas \(V\subseteq T\times T\) is the
vertical compatibility relation.}

The cycles satisfy the following containments. Parts (b) and (d)
will be important when comparing neighbouring tiles.

\begin{fact}[Column and row cycles]\label{fact:coordinate-cycles}
For every \(i\in\mathbb Z_m\) and \(j\in\mathbb Z_n\), the following hold:
\begin{enumerate}[label=\textup{(\alph*)},leftmargin=2.3em,itemsep=3pt]
\item \(Q_i\subseteq P_i\cup R^h\) and
      \(P_{i+1}\subseteq Q_i\cup S^h\).
\item \(Q_i[\sa{h}]\subseteq P_i\) and
      \(P_{i+1}[\nx{h}]\subseteq Q_i\).
\item \(W_j\subseteq V_j\cup R^v\) and
      \(V_{j+1}\subseteq W_j\cup S^v\).
\item \(W_j[\sa{v}]\subseteq V_j\) and
      \(V_{j+1}[\nx{v}]\subseteq W_j\).
\end{enumerate}
\end{fact}
\begin{proof}
Parts (a) and (c) follow from the defining properties of \(\Resp\),
including the returns to \(P_0\) and \(V_0\). For (b),
\(R^h[\sa{h}]=\varnothing\) and \(S^h[\nx{h}]=\varnothing\),
so Fact~\ref{fact:parts-union} gives the two containments.
Part (d) follows in the same way using \(R^v,S^v\).
\end{proof}

From now on both cycles are fixed, and all indices are read cyclically.
The position \((i,j)\in\mathbb Z_m\times\mathbb Z_n\)
corresponds to the pair \(P_i,V_j\). We put
\(C_{ij}=\Resp(P_i,V_j)\). By (M5),
\(\role(C_{ij})=\CellL{t}\) for some \(t\in T\).
This \(t\) is unique for the chosen \(C_{ij}\), by
Corollary~\ref{cor:roles} and the distinctness of the cell roles.
We define \(\tau(i,j)=t\), obtaining an assignment
\(\tau:\mathbb Z_m\times\mathbb Z_n\to T\).\footnote{Finiteness
and (H2) let us fix \(\Resp\) simultaneously for all demanded
pairs. Different choices may give different assignments, each of
which the argument below proves legal.}

\begin{figure}[!ht]
\centering
\begin{tikzpicture}[
  cycle set/.style={wm mid,font=\small,minimum width=1cm,
    minimum height=.6cm},
  cell set/.style={wm mid,font=\small,minimum width=1.75cm,
    minimum height=.65cm},
  response/.style={-{Stealth[length=1.7mm]},draw=black!60,
    line width=.55pt},
  heading/.style={font=\small}
]
  \node[heading] at (2.45,4.3) {Columns};
  \node[cycle set] (pi) at (.3,3.6) {$P_i$};
  \node[cycle set] (qi) at (2.05,3.6) {$Q_i$};
  \node[cycle set] (pin) at (3.8,3.6) {$P_{i+1}$};
  \node (hdots) at (5,3.6) {$\cdots$};
  \draw[response] (pi)--(qi);
  \draw[response] (qi)--(pin);
  \draw[response] (pin)--(hdots);

  \node[heading] at (8.85,4.3) {Rows};
  \node[cycle set] (vj) at (6.7,3.6) {$V_j$};
  \node[cycle set] (wj) at (8.45,3.6) {$W_j$};
  \node[cycle set] (vjn) at (10.2,3.6) {$V_{j+1}$};
  \node (vdots) at (11.4,3.6) {$\cdots$};
  \draw[response] (vj)--(wj);
  \draw[response] (wj)--(vjn);
  \draw[response] (vjn)--(vdots);

  \node[heading] at (4.65,2.4) {$P_i$};
  \node[heading] at (7.45,2.4) {$P_{i+1}$};
  \node[heading,anchor=east] at (2.95,1.55) {$V_{j+1}$};
  \node[heading,anchor=east] at (2.95,.4) {$V_j$};
  \draw[black!25] (3.25,-.15) rectangle (8.85,2.1);
  \draw[black!25] (6.05,-.15)--(6.05,2.1);
  \draw[black!25] (3.25,.975)--(8.85,.975);
  \node[cell set] at (4.65,1.55) {$C_{i,j+1}$};
  \node[cell set] at (7.45,1.55) {$C_{i+1,j+1}$};
  \node[cell set] at (4.65,.4) {$C_{ij}$};
  \node[cell set] at (7.45,.4) {$C_{i+1,j}$};
\end{tikzpicture}
\caption{The cycles supply the column and row indices. Arrows follow
chosen responses, with auxiliary inputs omitted. Each cell is
$C_{ij}=\Resp(P_i,V_j)$. Its role determines the tile at $(i,j)$.
All indices are cyclic.}
\label{fig:extracted-torus}
\end{figure}
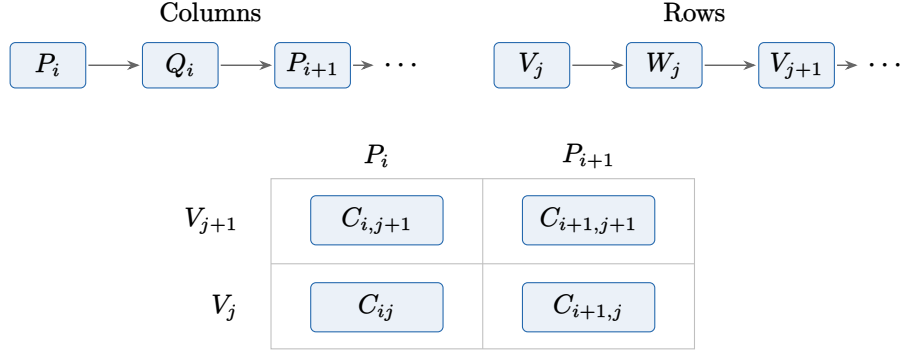

\begin{fact}[Parts of the cell responses]\label{fact:cell-parts}
For every \(i\in\mathbb Z_m\) and \(j\in\mathbb Z_n\),
\(C_{ij}\subseteq P_i\cup V_j\), and:
\begin{enumerate}[label=\textup{(\alph*)},leftmargin=2.3em,nosep]
\item \(C_{ij}[\sa{h}]\cup C_{ij}[\nx{h}]\subseteq P_i\);
\item \(C_{ij}[\sa{v}]\cup C_{ij}[\nx{v}]\subseteq V_j\).
\end{enumerate}
\end{fact}
\begin{proof}
The first containment follows from the definition of \(\Resp\).
By the support table and Lemma~\ref{lem:basic-concrete}(b),
\(V_j[\sa{h}]=V_j[\nx{h}]=\varnothing\), so
Fact~\ref{fact:parts-union} gives (a). The same argument gives (b),
using \(P_i[\sa{v}]=P_i[\nx{v}]=\varnothing\).
\end{proof}

We now check the horizontal matching rule for this assignment.
Consider neighbouring cells \(C_{ij},C_{i+1,j}\), and write
\(t=\tau(i,j)\), \(u=\tau(i+1,j)\).
Suppose, for a contradiction, that \((t,u)\notin H\).
Then (D6) places
\(\{\ExpL{h}{t},\NbrL{h}{u}\}\) in the demand set.
Our task is to obtain sets with these two roles from the chosen
cells and show that their required response cannot exist.
The intermediate Source and Step responses will produce the
Expected set from the left cell.

For each \(G\in\mathcal G\), we fix a nonempty set
\(F_G\subseteq K\) such that
\(\role(F_G)=\KeyL{G}\), called its
\emph{provider}.\footnote{Corollary~\ref{cor:roles} supplies these
sets.} The support table gives \(\MaxOf{\KeyL{G}}=\{\keyc{G}\}\).
For every \(k\in F_G\), (H0) assigns a Max-role to \(\{k\}\),
and (H1)(b) gives \(\KeyL{G}\leq\role(\{k\})\).
The only possible role is therefore \(\keyc{G}\).
Thus every element of \(F_G\) belongs to \(F_G[\keyc{G}]\), so
\(F_G=F_G[\keyc{G}]\). The whole set has role \(\KeyL{G}\),
while each singleton inside it has role \(\keyc{G}\).
Consequently, whenever \(A,F_G\) form a
demanded pair and \(\chi\ne\keyc{G}\),
Fact~\ref{fact:parts-union} gives
\(\Resp(A,F_G)[\chi]\subseteq A[\chi]\).
This will keep the indicated parts of a response inside its
cell or separator input.

We first obtain the Source set by pairing the left cell with
the provider specifically for \(\SrcL{h}{t}\).
Condition (D3) makes \(C_{ij},F_{\SrcL{h}{t}}\) a demanded pair,
so we put \(L^h_{ij}=\Resp(C_{ij},F_{\SrcL{h}{t}})\).
By (M6), \(\role(L^h_{ij})=\SrcL{h}{t}\).
The provider contributes only its private-key role, so
Fact~\ref{fact:cell-parts} gives
\(L^h_{ij}[\sa{h}]\subseteq P_i\) and
\(L^h_{ij}[\sa{v}]\subseteq V_j\).
The tile index \(t\) is inherited from the left cell's role.

We next pair the separator \(Q_i\) with the step provider.
By (D4), we can put \(U^h_i=\Resp(Q_i,F_{\StepL{h}})\), and
(M7) gives \(\role(U^h_i)=\StepL{h}\).
We have \(U^h_i[\nx{h}]\subseteq Q_i\).
Also \(U^h_i[\sa{h}]\subseteq Q_i[\sa{h}]\subseteq P_i\),
by Fact~\ref{fact:coordinate-cycles}(b).
So the \(\sa{h}\)-parts of both \(L^h_{ij}\) and \(U^h_i\)
lie in \(P_i\).

Condition (D5) lets us combine these sets:
we put \(O^h_{ij}=\Resp(L^h_{ij},U^h_i)\).
By (M9), its role is \(\SFL{h}\) or \(\ExpL{h}{t}\).
If \(\role(O^h_{ij})=\SFL{h}\), the support table and
Fact~\ref{fact:parts-union} give
\[
 O^h_{ij}=O^h_{ij}[\sa{h}]
 \subseteq L^h_{ij}[\sa{h}]\cup U^h_i[\sa{h}]
 \subseteq P_i.
\]
This contradicts Lemma~\ref{lem:basic-concrete}(c), since
\(\SFL{h}\) and \(\CPos\) are distinct Mid-roles.
Hence \(\role(O^h_{ij})=\ExpL{h}{t}\).
The support of \(\SrcL{h}{t}\) excludes \(\nx{h}\), so
\(O^h_{ij}[\nx{h}]\subseteq U^h_i[\nx{h}]\subseteq Q_i\).
The support of \(\StepL{h}\) excludes \(\sa{v}\), so
\(O^h_{ij}[\sa{v}]\subseteq L^h_{ij}[\sa{v}]\subseteq V_j\).
The Expected set therefore retains the left tile label \(t\)
and has the two parts needed for comparison with the right cell.

We put \(B^h_{i+1,j}=\Resp(C_{i+1,j},F_{\NbrL{h}{u}})\),
using (D3) for the right cell and its candidate provider.
By (M6), \(\role(B^h_{i+1,j})=\NbrL{h}{u}\).
Fact~\ref{fact:cell-parts} gives
\(B^h_{i+1,j}[\sa{v}]\subseteq V_j\).
For its other relevant part, the same fact and
Fact~\ref{fact:coordinate-cycles}(b) give
\(B^h_{i+1,j}[\nx{h}]\subseteq
P_{i+1}[\nx{h}]\subseteq Q_i\).
We have now constructed the two inputs required by (D6):

\begin{center}
\small
\renewcommand{\arraystretch}{1.15}
\begin{tabularx}{\linewidth}{@{}l l X X@{}}
\toprule
\(A\) & \(\role(A)\) & \(A[\nx{h}]\) & \(A[\sa{v}]\)\\
\midrule
\(O^h_{ij}\) & \(\ExpL{h}{t}\) & \(\subseteq Q_i\) & \(\subseteq V_j\)\\
\(B^h_{i+1,j}\) & \(\NbrL{h}{u}\) & \(\subseteq Q_i\) & \(\subseteq V_j\)\\
\bottomrule
\end{tabularx}
\end{center}

The common containing sets \(Q_i,V_j\) arise because the chosen
cells lie in consecutive columns and the same row.
Our assumption \((t,u)\notin H\), together with (D6) and (H2),
now gives \(Z=\Resp(O^h_{ij},B^h_{i+1,j})\).
This is a nonempty subset of \(O^h_{ij}\cup B^h_{i+1,j}\).
By (M11), \(\role(Z)=\NFL{h}\) or \(\role(Z)=\SFL{v}\).

If \(\role(Z)=\NFL{h}\), every singleton in \(Z\) has role
\(\nx{h}\). The table and Fact~\ref{fact:parts-union} therefore
give \(Z=Z[\nx{h}]\subseteq Q_i\).
This contradicts Lemma~\ref{lem:basic-concrete}(c), since
\(\role(Q_i)=\CSep\ne\NFL{h}\).
If \(\role(Z)=\SFL{v}\), the same argument gives
\(Z=Z[\sa{v}]\subseteq V_j\), again a contradiction, since
\(\role(V_j)=\RPos\ne\SFL{v}\).
Both permitted response roles are impossible. Hence \((t,u)\in H\).

The vertical check is the same argument \emph{mutatis mutandis}. We fix cells \(C_{ij},C_{i,j+1}\),
write \(t=\tau(i,j)\), \(u=\tau(i,j+1)\), and suppose
\((t,u)\notin V\). By (D3)--(D4), we put
\[
 L^v_{ij}=\Resp(C_{ij},F_{\SrcL{v}{t}}),\quad
 U^v_j=\Resp(W_j,F_{\StepL{v}}),\quad
 B^v_{i,j+1}=\Resp(C_{i,j+1},F_{\NbrL{v}{u}}).
\]
Rows (M6) and (M8) give
\(\role(L^v_{ij})=\SrcL{v}{t}\),
\(\role(U^v_j)=\StepL{v}\), and
\(\role(B^v_{i,j+1})=\NbrL{v}{u}\).
By Facts~\ref{fact:cell-parts} and~\ref{fact:coordinate-cycles}(d),
\(L^v_{ij}[\sa{v}]\subseteq V_j\) and
\(U^v_j[\sa{v}]\subseteq W_j[\sa{v}]\subseteq V_j\).
Also \(L^v_{ij}[\sa{h}]\subseteq P_i\) and
\(U^v_j[\nx{v}]\subseteq W_j\).

We put \(O^v_{ij}=\Resp(L^v_{ij},U^v_j)\), using (D5).
By (M10), its role is \(\SFL{v}\) or \(\ExpL{v}{t}\).
The first possibility would give
\(O^v_{ij}=O^v_{ij}[\sa{v}]\subseteq V_j\),
contrary to Lemma~\ref{lem:basic-concrete}(c).
Hence \(\role(O^v_{ij})=\ExpL{v}{t}\).
The supports and Fact~\ref{fact:parts-union} give
\(O^v_{ij}[\nx{v}]\subseteq W_j\) and
\(O^v_{ij}[\sa{h}]\subseteq P_i\).
For the Candidate set, Fact~\ref{fact:cell-parts} and
Fact~\ref{fact:coordinate-cycles}(d) give
\(B^v_{i,j+1}[\nx{v}]\subseteq
V_{j+1}[\nx{v}]\subseteq W_j\) and
\(B^v_{i,j+1}[\sa{h}]\subseteq P_i\).

Since \((t,u)\notin V\), (D7) and (H2) give
\(Z=\Resp(O^v_{ij},B^v_{i,j+1})\).
By (M12), \(\role(Z)=\NFL{v}\) or \(\role(Z)=\SFL{h}\).
The first case forces \(Z=Z[\nx{v}]\subseteq W_j\), and
the second forces \(Z=Z[\sa{h}]\subseteq P_i\).
Both contradict Lemma~\ref{lem:basic-concrete}(c), since
\(\role(W_j)=\RSep\) and \(\role(P_i)=\CPos\).
Therefore \((t,u)\in V\).

\begin{theorem}[Realization implies torus tiling]
\label{thm:realization-torus}
If \(\WM(D)\) has a finite realization, then \(D\) tiles a finite
torus.
\end{theorem}

\begin{proof}
The cycles give positive periods \(m,n\), and the cell responses
define \(\tau:\mathbb Z_m\times\mathbb Z_n\to T\).
The two checks establish Definition~\ref{def:torus}, including
the last column's and row's returns to index \(0\). They also
apply when \(m=1\) or \(n=1\), so the corresponding index
returns to itself.
\end{proof}
\section{From a torus tiling to a finite realization}
\label{sec:torus-to-realization}

We now fix a tiling
\(\tau:\mathbb Z_m\times\mathbb Z_n\to T\) of \(D\).
We write \(t_{ij}=\tau(i,j)\) and read the first index modulo \(m\)
and the second modulo \(n\). We construct a finite set \(K\) and a
partial role assignment satisfying (H0)--(H2) for \(\WM(D)\).

Each column, row and cell will be represented by a subset of \(K\).
We also need source, candidate, step, and expected sets, and finally
a response to every demanded pair.\footnote{This includes pairs
associated with unrelated positions.} We arrange the sets so that,
when positions fail the required comparison, their union contains
a response with role \(\SFL{h}\), \(\NFL{h}\), or, \emph{mutatis
mutandis}, its vertical counterpart. The construction uses different missing symbols to
make these containments hold.

Let us first list the sets we intend to construct. A \emph{state} is a
formal name for one of these sets. We let \(I\) be the finite set
of names in the table below, and write \(\lambda(a)\) for the role
listed beside \(a\). We will define a subset \(E_a\subseteq K\)
and assign \(\role(E_a)=\lambda(a)\). For instance, \(c_{ij}\)
is a name, \(E_{c_{ij}}\) is the set it names, and its prescribed
role is \(\lambda(c_{ij})=\CellL{t_{ij}}\).

An unbarred cell, source, candidate, or expected state refers to an
actual position of the tiling. A barred state is a \emph{seed}: an extra
representative ensuring that its tile's role exists even if that tile
never occurs in the chosen tiling. Indices range over
\(i\in\mathbb Z_m\), \(j\in\mathbb Z_n\), \(t\in T\), and
\(G\in\mathcal G\). All listed state names are distinct.

\begin{center}
\small
\renewcommand{\arraystretch}{1.2}
\begin{tabularx}{\linewidth}{@{}lX@{}}
\toprule
States & Roles\\\midrule
\(p_i,q_i\) & \(\CPos,\CSep\), respectively\\
\(v_j,w_j\) & \(\RPos,\RSep\), respectively\\
\(r_h,s_h,r_v,s_v\) & \(\NFL{h},\SFL{h},\NFL{v},\SFL{v}\), respectively\\
\(c_{ij},\bar c_t\) & \(\CellL{t_{ij}},\CellL{t}\), respectively\\
\(f_G\) & \(\KeyL{G}\)\\
\(\ell^h_{ij},b^h_{ij}\) & \(\SrcL{h}{t_{ij}},\NbrL{h}{t_{ij}}\), respectively\\
\(\ell^v_{ij},b^v_{ij}\) & \(\SrcL{v}{t_{ij}},\NbrL{v}{t_{ij}}\), respectively\\
\(\bar\ell^h_t,\bar b^h_t\) & \(\SrcL{h}{t},\NbrL{h}{t}\), respectively\\
\(\bar\ell^v_t,\bar b^v_t\) & \(\SrcL{v}{t},\NbrL{v}{t}\), respectively\\
\(u^h_i,u^v_j\) & \(\StepL{h},\StepL{v}\), respectively\\
\(o^h_{ij},o^v_{ij}\) & \(\ExpL{h}{t_{ij}},\ExpL{v}{t_{ij}}\), respectively\\
\(\bar o^h_t,\bar o^v_t\) & \(\ExpL{h}{t},\ExpL{v}{t}\), respectively\\
\bottomrule
\end{tabularx}
\end{center}
Every Mid-role occurs in the table, so \(\lambda:I\to\Mid_D\)
is surjective.

Choose pairwise distinct formal symbols and put
\begin{equation}
\begin{split}
 J&=\{x_i:i\in\mathbb Z_m\}\,\dot\cup\,
       \{y_j:j\in\mathbb Z_n\}\,\dot\cup\,
       \{z\}\,\dot\cup\,\{z_t:t\in T\},\\
 K&=J\times\MaxPts_D.
\end{split}
\label{eq:coordinate-carrier}
\end{equation}
A carrier point is a pair \((u,\chi)\): a marker and a Max-role.
For each maximal role \(\chi\), the set \(J\times\{\chi\}\) is its
\emph{block}. The marker \(x_i\) records torus column \(i\),
and \(y_j\) records torus row \(j\). The marker \(z\) helps
separate the column and row sets from their separator sets. We reserve
\(z_t\) for the seeds of tile \(t\). The symbols are also distinct
from the state names.\footnote{For \(m=1\) or \(n=1\),
cyclic indices can coincide within a marker family. Different
families remain disjoint.}

We write \(S_a=\MaxOf{\lambda(a)}\) for the support prescribed
by \(\lambda(a)\). Thus \(E_a\) must meet exactly the blocks
\(J\times\{\chi\}\) with \(\chi\in S_a\). For each \(\chi\in S_a\), we specify a \emph{blocklist}
\(B_\chi(a)\subseteq J\), listing the symbols omitted from that
block. We then define
\begin{equation}
 E_a=\{(u,\chi):\chi\in S_a,\ u\in J\setminus B_\chi(a)\}.
\label{eq:compressed-incidence}
\end{equation}
Thus \(B_\chi(a)=\varnothing\) includes the entire supported block.
A block outside \(S_a\) contributes no points at all.

To see how this records coordinates, consider two sets that contain
all of one block except the points marked \(x_i\) and \(x_k\),
respectively. Their union fills that block iff \(i\ne k\):
\begin{equation}
 (J\setminus\{x_i\})\cup(J\setminus\{x_k\})=J
 \quad\Longleftrightarrow\quad i\ne k.
\label{eq:marker-test}
\end{equation}
If the symbols agree, that point remains missing from both inputs.
For a next-column test, we use \(x_{k-1}\) in place of \(x_k\).
The missing symbols then agree iff \(k=i+1\). So the
same test can compare either two columns or a column and its intended
successor. The following lemma also handles blocks that
only one input supports.

\begin{lemma}[How to check a response containment]
\label{lem:response-columns}
Suppose every supported block of \(E_d\) is populated. Then
\(E_d\subseteq E_a\cup E_b\) holds iff
\(S_d\subseteq S_a\cup S_b\) and, for every \(\chi\in S_d\),
\begin{equation}
\begin{array}{c|c}
\text{inputs supporting }\chi & \text{required inclusion}\\\hline
 a\text{ only} & B_\chi(a)\subseteq B_\chi(d)\\
 b\text{ only} & B_\chi(b)\subseteq B_\chi(d)\\
 a\text{ and }b & B_\chi(a)\cap B_\chi(b)\subseteq B_\chi(d).
\end{array}
\label{eq:column-test}
\end{equation}
\end{lemma}

\begin{proof}
We first check that each block used by \(E_d\) occurs in at least
one input. Otherwise a point of \(E_d\) in that block would lie
outside \(E_a\cup E_b\). Within a covered block, the symbols
missing from the input union are \(B_\chi(a)\), \(B_\chi(b)\),
or \(B_\chi(a)\cap B_\chi(b)\), according to the three cases.
Containment holds iff all those symbols are also missing from
\(E_d\). These are the displayed inclusions.
\end{proof}

In particular, when both inputs support a block, their union fills
it iff their blocklists have empty intersection. We use this
test for the mismatch responses.

We now specify every supported blocklist. To \emph{copy} a list
from another state means to use the same subset of \(J\) in that block.

For the position and separator states, the horizontal lists are
\begin{equation}
\begin{array}{c|ccc}
 a & B_{\nx{h}}(a)&B_{\sa{h}}(a)&B_{\tg{h}}(a)\\\hline
 p_i&\{x_{i-1},z\}&\{x_i\}&\{z\}\\
 q_i&\{x_i\}&\{x_i,z\}&\{z\}.
\end{array}
\label{eq:generators}
\end{equation}
For \(v_j,w_j\), replace \(p_i,q_i,x_i,h\) by
\(v_j,w_j,y_j,v\), respectively. The lists for \(q_i,p_i\) both contain \(x_i\) at \(\sa{h}\),
and those for \(q_i,p_{i+1}\) both contain \(x_i\) at \(\nx{h}\).
These shared omissions will give the required column sequence.
The extra \(z\) makes column and separator sets incomparable
by inclusion (Lemma~\ref{lem:compressed-separation}).

We next specify the mismatch and provider states.
Each of \(r_h,s_h,r_v,s_v,f_G\) has an empty blocklist in its unique
supported block. Thus \(E_{r_h}=J\times\{\nx{h}\}\) and
\(E_{s_h}=J\times\{\sa{h}\}\), with the corresponding vertical
formulas for \(r_v,s_v\). Each provider fills its key block:
\(E_{f_G}=J\times\{\keyc{G}\}\). 

We give the cell at \((i,j)\) the same horizontal lists as column
\(i\) and the same vertical lists as row \(j\):
\begin{equation}
\begin{aligned}
 B_{\nx{h}}(c_{ij})&=\{x_{i-1},z\}, &
 B_{\sa{h}}(c_{ij})&=\{x_i\},\\
 B_{\nx{v}}(c_{ij})&=\{y_{j-1},z\}, &
 B_{\sa{v}}(c_{ij})&=\{y_j\}.
\end{aligned}
\label{eq:cells}
\end{equation}
The seed cell has \(B_\chi(\bar c_t)=\{z_t\}\) in all four
coordinate blocks.

We define the source and candidate sets by copying two blocklists
from the cell and adding the private-key block prescribed by their
respective roles.
For an actual cell we copy from \(c_{ij}\), and for a seed we copy
from \(\bar c_t\). The private-key list is always \(\{z\}\).
The support table determines which two coordinate blocks to copy:
\(\ell^h\) uses \(\sa{h},\sa{v}\), and \(b^h\) uses
\(\nx{h},\sa{v}\). Vertically, \(\ell^v\) uses
\(\sa{v},\sa{h}\), and \(b^v\) uses \(\nx{v},\sa{h}\).
For instance, the lists of \(\ell^h_{ij}\) are \(\{x_i\}\) at
\(\sa{h}\), \(\{y_j\}\) at \(\sa{v}\), and \(\{z\}\)
at its private key. The tile index determines this private key, whereas
the coordinates determine the two coordinate lists.

We obtain the step \(u^h_i\) by copying the
\(\nx{h}\)- and \(\sa{h}\)-lists of \(q_i\), and using
\(\{z\}\) at its private key. So its coordinate lists are \(\{x_i\}\) and
\(\{x_i,z\}\), respectively. The vertical step \(u^v_j\) copies
\(w_j\) in the same way.

Now, the expected set for the tile at \((i,j)\) will be compared with
a candidate for its neighbour. Horizontally, the missing marker
\(x_i\) will match a candidate from column \(i+1\), whose
next-column list contains \(x_{(i+1)-1}=x_i\). The missing marker
\(y_j\) will check that the row is still \(j\). We therefore set
\begin{equation}
\begin{aligned}
 B_{\nx{h}}(o^h_{ij})&=\{x_i\}, &
 B_{\sa{v}}(o^h_{ij})&=\{y_j\},\\
 B_{\nx{v}}(o^v_{ij})&=\{y_j\}, &
 B_{\sa{h}}(o^v_{ij})&=\{x_i\}.
\end{aligned}
\label{eq:expected}
\end{equation}
Both private-key lists of each actual expected state are \(\{z\}\).

For the expected seeds, in every one of their four
supported blocks, put \(B_\chi(\bar o^h_t)=\{y_0\}\) and
\(B_\chi(\bar o^v_t)=\{x_0\}\).
In the next-column block, a horizontal expected seed omits \(y_0\).
Every horizontal candidate omits only column symbols, \(z\), or
seed symbols. These blocklists are disjoint, so the union of the
represented sets fills the block and supplies the mismatch response. The vertical check
uses the next-row block.

\begin{lemma}[Sizes and exact supports]
\label{lem:compressed-support}
Each \(E_a\) is a proper subset of \(K\) with at least two points.
The second components of its elements are exactly the roles in \(S_a\).
\end{lemma}
\begin{proof}
We have \(|J|=m+n+|T|+1\geq4\), and every supported blocklist has
at most two elements. So each supported block contains at least
two points. Every Mid-support is nonempty, so \(|E_a|\geq2\).
By definition, the second components occurring in \(E_a\) are
exactly \(S_a\).

Every Mid-support has at most four singleton roles, whereas
\(|\MaxPts_D|\geq12\). A singleton role outside \(S_a\) supplies a point
of \(K\setminus E_a\). Hence \(E_a\subsetneq K\).
\end{proof}

\begin{lemma}[Separation of different roles]
\label{lem:compressed-separation}
If \(\lambda(a)\ne\lambda(b)\), then
\(E_a\nsubseteq E_b\) and \(E_b\nsubseteq E_a\).
\end{lemma}
\begin{proof}
Every support has size one, three, or four. No three-element support
is properly contained in a four-element support: a generator uses a
tag, which no four-element support has. A keyed three-element support
cannot lie in a cell support. The source or step sharing a key
with an expected role needs the same-axis coordinate missing from
that expected support. Candidate keys occur in no expected support.

Equal nonsingleton supports of different roles occur only in
\(\{\CPos,\CSep\}\), in \(\{\RPos,\RSep\}\), and among cell
roles. Private keys distinguish the other supports. Each singleton
support belongs to exactly one Mid-role. We can therefore check
separation in the following four cases.

First, suppose the supports are incomparable. We choose
\(\chi\in S_a\setminus S_b\). A point in the populated
\(\chi\)-block of \(E_a\) lies outside \(E_b\).
A singleton role in \(S_b\setminus S_a\) gives the reverse witness.

Second, suppose \(S_a=\{\chi\}\subsetneq S_b\). The state \(a\) is a
mismatch or provider, so \(E_a=J\times\{\chi\}\).
Every supported list of the nonsingleton-support state \(b\) is
nonempty. We choose \(u\in B_\chi(b)\). The point \((u,\chi)\) belongs
to \(E_a\setminus E_b\).
A point in any other supported block of \(E_b\) lies outside \(E_a\).

Third, we compare a column set with a column separator set.
For every \(i,k\), the point \((z,\sa{h})\) belongs to
\(E_{p_i}\setminus E_{q_k}\), while \((z,\nx{h})\) belongs to
\(E_{q_k}\setminus E_{p_i}\).
The vertical witnesses are \((z,\sa{v})\) and \((z,\nx{v})\).

Finally, we compare sets with different cell roles.
If actual cells \(c_{ij},c_{k\ell}\) carry different tiles, their
coordinates differ. If \(i\ne k\), the point \((x_k,\sa{h})\) lies in the first set difference,
and \((x_i,\sa{h})\) lies in the reverse difference.
If \(i=k\), then \(j\ne\ell\), and the corresponding witnesses
are \((y_\ell,\sa{v})\) and \((y_j,\sa{v})\).
For an actual cell \(c_{ij}\) and a seed \(\bar c_t\), we use
\((z_t,\sa{h})\) in the first difference and \((x_i,\sa{h})\)
in the reverse difference.
For two different seeds \(\bar c_t,\bar c_u\), we use
\((z_u,\sa{h})\) in the first difference and
\((z_t,\sa{h})\) in the second.
\end{proof}

\begin{proposition}[Assigning the roles]
\label{prop:compressed-roles}
Put \(\role(K)=r\), \(\role(\{(u,\chi)\})=\chi\) for every
\((u,\chi)\in K\), and \(\role(E_a)=\lambda(a)\) for every
\(a\in I\). Leave all other nonempty subsets unassigned. These prescriptions define
a partial function satisfying (H0) and (H1).
\end{proposition}
\begin{proof}
Sizes and properness separate each \(E_a\) from the root set and
all singletons. If \(E_a=E_b\), the separation lemma forces
\(\lambda(a)=\lambda(b)\). Thus no set receives conflicting roles,
and (H0) holds.

For (H1), consider the three kinds of assigned sets. The root \(K\)
contains a representative of every Mid-role by the state table, and a
singleton of every Max-role by the definition of \(K\). It also
contains itself. These are exactly all roles above \(r\).

Inside \(E_a\), the set itself witnesses \(\lambda(a)\), and its
populated blocks supply singletons of each maximal role in \(S_a\).
These are exactly the roles above \(\lambda(a)\). Conversely, any
assigned set contained in \(E_a\) is either a singleton of a supported
singleton role or some \(E_b\). In the latter case separation forces
\(\lambda(b)=\lambda(a)\). The root set is excluded by properness.
This proves both directions of (H1) at a Mid-set.

Finally, a singleton contains only itself as a nonempty subset, and
its assigned role is maximal. This proves (H1) there as well.
\end{proof}

It remains to check (H2). We must provide a response for every pair
of assigned sets whose roles form a demand. The two sets need not
refer to matching coordinates, which makes the check a bit more tricky.
For example, a horizontal source at
\((i,j)\) can be paired with a step for any column \(k\). We must
also handle seeds.

For each pair of states \(a,b\) with
\(\{\lambda(a),\lambda(b)\}\in\D_D\), we choose a state
\(d=\rho(a,b)\) and prove that \(E_d\) is a response to
\(E_a,E_b\). We give the choices in the order (D1)--(D7) of
Definition~\ref{def:wmpair}, and set \(\rho(b,a)=\rho(a,b)\).
We check \(E_d\subseteq E_a\cup E_b\) using
Lemma~\ref{lem:response-columns}. Each chosen output has a Mid-role
different from both input roles, so
\(\role(E_d)=\lambda(d)\notin\up\lambda(a)\cup\up\lambda(b)\).
Nonemptiness follows from Lemma~\ref{lem:compressed-support}.

For (D1), we make the column and row sequences follow the fixed
torus indices. We define
\[
 \rho(p_i,r_h)=q_i,\qquad \rho(q_i,s_h)=p_{i+1},
 \qquad
 \rho(v_j,r_v)=w_j,\qquad \rho(w_j,s_v)=v_{j+1}.
\]
For \(\rho(p_i,r_h)=q_i\), the set \(E_{r_h}\) supplies the
whole \(\nx{h}\)-block. The other output blocks come from
\(E_{p_i}\): the required blocklist inclusions are
\(\{x_i\}\subseteq\{x_i,z\}\) at \(\sa{h}\) and
\(\{z\}\subseteq\{z\}\) at the tag. Thus
\(E_{q_i}\subseteq E_{p_i}\cup E_{r_h}\).

For \(\rho(q_i,s_h)=p_{i+1}\), the helper fills the
\(\sa{h}\)-block. At \(\nx{h}\) the inclusion is
\(\{x_i\}\subseteq\{x_i,z\}\), and the tag list is unchanged.
This gives \(E_{p_{i+1}}\subseteq E_{q_i}\cup E_{s_h}\).
Replacing \(x_i,h\) by \(y_j,v\) proves the two vertical responses
by the same three block checks.

For (D2), the inputs are \(E_{p_i},E_{v_j}\), with roles
\(\CPos,\RPos\). We put \(\rho(p_i,v_j)=c_{ij}\), so the
proposed response is \(E_{c_{ij}}\), with
\(\role(E_{c_{ij}})=\CellL{t_{ij}}\). Here \(t_{ij}\) is
the tile assigned by our fixed tiling.
The inputs have disjoint supports. The cell's two horizontal lists
are copied from \(p_i\), and its two vertical lists from \(v_j\). Each required one-input inclusion is therefore an equality.
Thus \(E_{c_{ij}}\subseteq E_{p_i}\cup E_{v_j}\).

For (D3), the inputs are a cell set with role \(\CellL{t}\)
and a provider set \(E_{f_G}\), with
\(\role(E_{f_G})=\KeyL{G}\). We need a response with role
\(G\). For \(t=t_{ij}\), we define
\[
\begin{aligned}
 \rho(c_{ij},f_{\SrcL{h}{t}})&=\ell^h_{ij},&
 \rho(c_{ij},f_{\NbrL{h}{t}})&=b^h_{ij},\\
 \rho(c_{ij},f_{\SrcL{v}{t}})&=\ell^v_{ij},&
 \rho(c_{ij},f_{\NbrL{v}{t}})&=b^v_{ij}.
\end{aligned}
\]
For seed cells, we replace \(c_{ij}\) by \(\bar c_t\) and each
output by its barred version indexed by \(t\).

Each of the two output blocks copied from the cell is supported
by that input alone, so its required inclusion is an equality.
The private-key block is supplied by the provider alone, with
required inclusion \(\varnothing\subseteq\{z\}\). This proves
the response containments for all four roles, both for actual cells
and for seeds.

For (D4), we use \(E_{u^h_i}\) as a response to
\(E_{q_i},E_{f_{\StepL{h}}}\), and \(E_{u^v_j}\) as a response
to \(E_{w_j},E_{f_{\StepL{v}}}\). Thus
\(\rho(q_i,f_{\StepL{h}})=u^h_i\) and
\(\rho(w_j,f_{\StepL{v}})=u^v_j\).
The two blocklists copied from the separator give equalities in the
containment test. The provider supplies the key block, again with
inclusion \(\varnothing\subseteq\{z\}\).

For (D5), the horizontal inputs are \(E_{\ell^h_{ij}}\)
and \(E_{u^h_k}\), with roles \(\SrcL{h}{t_{ij}}\) and
\(\StepL{h}\). When \(k=i\), we use the expected set
\(E_{o^h_{ij}}\). When \(k\ne i\), the two sets together fill
the same-column block and give a response of role \(\SFL{h}\).
The vertical choice follows the same rule for rows. We define
\begin{equation}
\begin{aligned}
 \rho(\ell^h_{ij},u^h_k)&=
 \begin{cases}o^h_{ij},&k=i,\\s_h,&k\ne i,\end{cases}
 &
 \rho(\ell^v_{ij},u^v_k)&=
 \begin{cases}o^v_{ij},&k=j,\\s_v,&k\ne j.\end{cases}
\end{aligned}
\label{eq:first-response}
\end{equation}
For seed sources, we put
\(\rho(\bar\ell^h_t,u^h_k)=s_h\) and
\(\rho(\bar\ell^v_t,u^v_k)=s_v\).

For the aligned horizontal output \(o^h_{ij}\), write \(t=t_{ij}\).
The complete block check is
\[
\begin{array}{c|c|c}
\text{output block}&\text{sole supporting input}&
 B_\chi(\text{input})=B_\chi(o^h_{ij})\\\hline
\nx{h}&u^h_i&\{x_i\}\\
\sa{v}&\ell^h_{ij}&\{y_j\}\\
\keyc{\SrcL{h}{t}}&\ell^h_{ij}&\{z\}\\
\keyc{\StepL{h}}&u^h_i&\{z\}.
\end{array}
\]
Every required inclusion is an equality. Notice that the common
input coordinate \(\sa{h}\) is absent from the expected output.
For an aligned vertical output, the four corresponding blocks are
\(\nx{v},\sa{h},\keyc{\SrcL{v}{t}},\keyc{\StepL{v}}\), with lists
\(\{y_j\},\{x_i\},\{z\},\{z\}\).

For the horizontal mismatch branch, the proposed output is
\(E_{s_h}=J\times\{\sa{h}\}\). Both inputs support this singleton role,
and their blocklists have intersection
\begin{equation}
 B_{\sa{h}}(\ell^h_{ij})\cap B_{\sa{h}}(u^h_k)
 =\{x_i\}\cap\{x_k,z\}=\varnothing
 \quad\text{when }i\ne k.
\label{eq:first-mismatch}
\end{equation}
So every point of that block is present in at least one input.
The output list is empty, so \eqref{eq:column-test} holds.
For a seed source the intersection is
\(\{z_t\}\cap\{x_k,z\}=\varnothing\), for every \(k\).
Vertically the intersections are
\(\{y_j\}\cap\{y_k,z\}\) and \(\{z_t\}\cap\{y_k,z\}\).
They are empty in exactly the required cases.

For (D6), we consider an expected set for tile \(t\) and a
candidate set for tile \(u\), where \((t,u)\notin H\).
Because the fixed tiling is legal, actual occurrences of these tiles
cannot be horizontal neighbours in that order. Their column or row
coordinates must therefore fail the neighbour test. We use that
failure to obtain a response. If either input is a seed, we choose
\(r_h\). For two actual inputs, we define
\begin{equation}
 \rho(o^h_{ij},b^h_{k\ell})=
 \begin{cases}
 r_h,&k\ne i+1,\\
 s_v,&k=i+1.
 \end{cases}
\label{eq:forbidden-h}
\end{equation}
Here the demanded tiles are \(t=t_{ij}\), \(u=t_{k\ell}\).
The two common coordinate blocks have blocked intersections
\begin{align}
 B_{\nx{h}}(o^h_{ij})\cap B_{\nx{h}}(b^h_{k\ell})
 &=\{x_i\}\cap\{x_{k-1},z\},
 \label{eq:next-test}\\
 B_{\sa{v}}(o^h_{ij})\cap B_{\sa{v}}(b^h_{k\ell})
 &=\{y_j\}\cap\{y_\ell\}.
 \label{eq:same-test}
\end{align}
The first is nonempty iff \(k=i+1\), and the second
iff \(\ell=j\).

If \(k\ne i+1\), \eqref{eq:next-test} is empty and the union of
the inputs contains the entire \(\nx{h}\)-block. This is
\(E_{r_h}\), proving the first branch.
If \(k=i+1\), then \(\ell\ne j\): otherwise the two input tiles
would lie on the actual horizontal edge
from \((i,j)\) to \((i+1,j)\), and legality of \(\tau\) would
put their pair in \(H\). Consequently \eqref{eq:same-test} is empty,
and the union contains \(E_{s_v}\). This proves the second branch.

For seed cases, the next-column lists are
\[
\begin{array}{c|cc}
 &\text{actual}&\text{seed}\\\hline
\text{expected input}&\{x_i\}&\{y_0\}\\
\text{candidate input}&\{x_{k-1},z\}&\{z_u\}.
\end{array}
\]
Whenever at least one input is a seed, these lists are disjoint
by the distinctness of the symbols. Hence the union of the two
inputs contains the whole \(\nx{h}\)-block \(E_{r_h}\),
proving all three seed cases.


We verify (D7) using the corresponding vertical comparison.
For \((t,u)\notin V\), we choose \(r_v\) whenever either input
is a seed. For actual inputs, we define
\[
 \rho(o^v_{ij},b^v_{k\ell})=
 \begin{cases}
 r_v,&\ell\ne j+1,\\
 s_h,&\ell=j+1.
 \end{cases}
\]
The two tests are
\[
 \{y_j\}\cap\{y_{\ell-1},z\},\qquad
 \{x_i\}\cap\{x_k\}.
\]
The first is empty when \(\ell\ne j+1\), giving \(E_{r_v}\).
If \(\ell=j+1\), legality and the forbidden tile pair force
\(k\ne i\). The second intersection is then empty, giving
\(E_{s_h}\). For a seed expected input we use \(\{x_0\}\), which is
disjoint from every actual candidate list \(\{y_{\ell-1},z\}\)
and every seed candidate list \(\{z_u\}\). If only the candidate
is a seed, its list \(\{z_u\}\) is disjoint from \(\{y_j\}\).
This verifies all vertical seed cases.

The seven families cover every demanded state pair, including seeds.
Legality of \(\tau\) was needed in (D6) and (D7): a forbidden
tile pair cannot occur at positions with both coordinates required
for adjacency.

\begin{theorem}[Compressed construction]
\label{thm:torus-realization}
If \(D=(T,H,V)\) has a legal \(m\)-by-\(n\) torus tiling, then
\(\WM(D)\) has a realization on
\[
 K=J\times\MaxPts_D,\qquad
 \boxed{|K|=(m+n+|T|+1)(4|T|+8).}
\]
\end{theorem}
\begin{proof}
Proposition~\ref{prop:compressed-roles} gives (H0) and (H1).
For (H2), we take any assigned sets \(X,Y\) whose roles form a
demand. Only sets of the form \(E_a\) carry Mid-roles, so we choose states
\(a,b\) with \(X=E_a\) and \(Y=E_b\).\footnote{Several states
may name the same set. Any names work because every demanded state
pair was checked. Condition (H1)(b) also permits sets with the same
role to contain one another.}
The choices above supply a state \(d\) such that
\(\varnothing\ne E_d\subseteq E_a\cup E_b=X\cup Y\) and
\(\role(E_d)\notin\up\role(X)\cup\up\role(Y)\).
Hence \(E_d\) is the required response.
The cardinality follows from \eqref{eq:coordinate-carrier}.
\end{proof}

\section{Recognizing a realization by a formula}
\label{sec:recognition}

We fix \(D\) and \((P,\D)=(P_D,\D_D)\). We construct a
formula from this pair alone that fails on \(F(K)\) iff
the pair has a realization on \(K\).

Recall that the successors of a world \(X\) in \(F(K)\) are
its nonempty subsets.

\Needspace{10\baselineskip}
A \emph{valuation} assigns each variable \(p\) a set \(v(p)\) of
worlds.  It must be persistent: if \(X\in v(p)\) and
\(\varnothing\neq Y\subseteq X\), then \(Y\in v(p)\).
Forcing under this valuation is defined recursively:
\begin{enumerate}[label=\textup{(\roman*)},leftmargin=2.3em,nosep]
\item \(X\Vdash p\) iff \(X\in v(p)\).
\item No world forces \(\bot\).
\item \(X\Vdash\varphi\wedge\psi\) iff it forces both
conjuncts.
\item \(X\Vdash\varphi\vee\psi\) iff it forces at least
one disjunct at \(X\) itself.
\item \(X\Vdash\varphi\to\psi\) iff every nonempty
\(Y\subseteq X\) forcing \(\varphi\) also forces \(\psi\).
\end{enumerate}
Consequently, \(X\Vdash\neg\varphi\) means that no nonempty
subset of \(X\) forces \(\varphi\).  Every world forces \(\top\).
Persistence follows by induction: in the implication step, every
successor of \(Y\subseteq X\) is also a successor of \(X\).

Failure of forcing, written \(X\not\Vdash\varphi\), is weaker than
forcing \(\neg\varphi\): a variable may fail at \(X\) and hold at
a proper subset.  We repeatedly use the following direct consequence
of the implication clause:
\begin{equation}
 X\not\Vdash\varphi\to\psi
 \quad\Longleftrightarrow\quad
 \text{some nonempty }Y\subseteq X\text{ forces }\varphi
 \text{ and fails }\psi.
\label{eq:failure-implication}
\end{equation}
Recall that \(F(K)\not\models\varphi\) supplies a valuation
and a world failing \(\varphi\).

The idea is to introduce one variable \(e_\lambda\) for each role
\(\lambda\in P\). Given a realization, we will define
\(X\Vdash e_\lambda\) to mean that no nonempty \(E\subseteq X\)
has \(\role(E)=\lambda\). This valuation is persistent: if
no such set exists inside \(X\), none exists inside a subset
of \(X\). So \(X\not\Vdash e_\lambda\) means that some
\(E\subseteq X\) has \(\role(E)=\lambda\), even when \(X\)
itself has no assigned role. We call \(\lambda\) \emph{visible
at \(X\)} when \(X\not\Vdash e_\lambda\).
For the converse direction, we will use the formula clauses to recover
a role assignment from these failures of forcing.

We put \(N:=\bigwedge_{\lambda\in P}e_\lambda\), and define
\begin{equation}
 A_\lambda:=\bigwedge_{\delta\notin\up\lambda}e_\delta,
 \qquad
 A_{\lambda,\gamma}:=
 \bigwedge_{\delta\notin\up\lambda\cup\up\gamma}e_\delta.
\label{eq:recognition-A}
\end{equation}
Thus \(N\) says that no role is visible, and \(A_\lambda\) says
that all visible roles lie above \(\lambda\).  It does not assert that
\(\lambda\) is visible.  The formula \(A_{\lambda,\gamma}\) says that every visible role
lies above \(\lambda\) or above \(\gamma\).

We let \(\Gamma(P,\D)\) be the conjunction of the following
clauses.
\begin{description}[leftmargin=3.0em,style=nextline,labelsep=.6em]
\item[(N)] \(\neg N\).  Some role must be visible at every successor.
This includes singleton worlds.

\item[(O)] \(e_\lambda\to e_r\) for each \(\lambda\in\Mid\),
and \(e_\chi\to e_\lambda\) whenever
\(\chi\in\MaxOf{\lambda}\).  Visibility is thereby closed upwards:
if \(e_\lambda\) fails and \(e_\chi\to e_\lambda\) holds, then
\(e_\chi\) fails too.

\item[(B)] \((A_\lambda\to e_\lambda)\to e_\lambda\) for every
\(\lambda\neq r\).  In other words, if \(e_\lambda\) fails, the inner
implication must fail.  Equivalence~\eqref{eq:failure-implication} supplies a successor
where \(A_\lambda\) holds and \(e_\lambda\) still fails.  Together
with (O), this makes the visible roles exactly \(\up\lambda\).

\item[(S)] \(A_{\lambda,\gamma}\to(e_\lambda\vee e_\gamma)\)
for each demand \(\{\lambda,\gamma\}\in\D\).  When both input roles
are visible, some visible role must lie outside their two upsets.
That role will provide a legal response.
\end{description}
Finally, we put \(\alpha(P,\D):=\Gamma(P,\D)\to e_r\).
Refuting this formula supplies a world satisfying all the clauses while
the root remains visible.  The formula has \(|P|\) variables and size
polynomial in \(|P|+|\D|\): each clause uses at most linearly many
variables, and there are at most quadratically many order clauses.

Suppose first that a realization on \(K\) is given. From it we
build a valuation that refutes the formula. To record which roles occur inside a world,
we define its \emph{trace} by
\begin{equation}
 \tr(X):=\{\lambda\in P:
 \text{some nonempty }E\subseteq X\text{ has }\role(E)=\lambda\}.
\label{eq:old-4-15}
\end{equation}
\begin{lemma}[Properties of traces]
\label{lem:trace}
Every trace is a nonempty upset.  If
\(\varnothing\neq Y\subseteq X\), then
\(\tr(Y)\subseteq\tr(X)\).  If \(\role(E)=\lambda\), then
\(\tr(E)=\up\lambda\).
\end{lemma}

\begin{proof}
Every world \(X\) contains a singleton, which has a Max-role by
(H0), so \(\tr(X)\neq\varnothing\). If \(\lambda\in\tr(X)\)
and \(\lambda\leq\gamma\), we choose \(E\subseteq X\) with
\(\role(E)=\lambda\). By (H1)(a), some nonempty \(C\subseteq E\)
has \(\role(C)=\gamma\), so \(\gamma\in\tr(X)\). Thus
\(\tr(X)\) is an upset. Every assigned subset of \(Y\subseteq X\)
is also a subset of \(X\), giving \(\tr(Y)\subseteq\tr(X)\).
Finally, if \(\role(E)=\lambda\), (H1)(a) gives
\(\up\lambda\subseteq\tr(E)\), and (H1)(b) gives the reverse
inclusion.
\end{proof}

It turns out that the response condition can be read off the traces.

\begin{lemma}[Responses as forbidden traces]
\label{lem:forbidden-trace}
For an assignment satisfying (H0) and (H1), the response condition for
a demand \(\{\lambda,\gamma\}\) is equivalent to the absence of a
world with trace \(\up\lambda\cup\up\gamma\).
\end{lemma}

\begin{proof}
Suppose the response condition holds, and assume for contradiction
that \(\tr(X)=\up\lambda\cup\up\gamma\). We choose
\(E,F\subseteq X\) with \(\role(E)=\lambda\) and
\(\role(F)=\gamma\). By (H2), there is a response
\(R\subseteq E\cup F\). Since \(R\subseteq X\), its role belongs
to \(\tr(X)\), but a response must have its role outside
\(\up\lambda\cup\up\gamma\). This is a contradiction.

Conversely, suppose no world has that trace. We take any sets
\(E,F\) with \(\role(E)=\lambda\) and \(\role(F)=\gamma\).
Lemma~\ref{lem:trace} gives
\(\up\lambda\cup\up\gamma\subseteq\tr(E\cup F)\).
Equality is excluded, so we choose
\(\delta\in\tr(E\cup F)\setminus(\up\lambda\cup\up\gamma)\).
By the definition of trace, some nonempty \(R\subseteq E\cup F\)
has \(\role(R)=\delta\). This is the response required by (H2).
\end{proof}

We now define the valuation by \(X\Vdash e_\lambda\) iff
\(\lambda\notin\tr(X)\).  Trace monotonicity proves persistence.
For any \(U\subseteq P\), \(X\) forces the conjunction of
variables indexed outside \(U\) iff \(\tr(X)\subseteq U\).

\begin{proposition}[Forward recognition]
\label{prop:recognition-forward}
A realization on \(K\) yields
\(F(K)\not\models\alpha(P,\D)\).
\end{proposition}

\begin{proof}
Under the preceding valuation, every world forces all four clause
families, as follows.

For (N), every trace is nonempty, so no world forces \(N\).  Hence every world
forces \(\neg N\).

For (O), upward closure of traces makes forcing \(e_\gamma\) entail forcing
\(e_\lambda\) whenever \(\lambda\leq\gamma\).  Since this holds
at all worlds, the required implications hold everywhere.

For (B), suppose \(Y\Vdash A_\lambda\to e_\lambda\) but
\(Y\not\Vdash e_\lambda\). Then \(\lambda\in\tr(Y)\), so we
choose \(E\subseteq Y\) with \(\role(E)=\lambda\).
By Lemma~\ref{lem:trace}, \(\tr(E)=\up\lambda\). Hence
\(E\Vdash A_\lambda\) and \(E\not\Vdash e_\lambda\),
contradicting the implication at \(Y\). Thus every world forcing
the inner implication forces its consequent, which proves (B).

For (S), suppose \(Y\Vdash A_{\lambda,\gamma}\) while both input variables
fail there. Its trace lies in \(\up\lambda\cup\up\gamma\).
Visibility of both roles and upward closure give the reverse inclusion.  This contradicts
Lemma~\ref{lem:forbidden-trace}.  Thus at least one input variable
is forced whenever the antecedent is forced, proving (S).

Finally, \(\role(K)=r\) makes \(r\) visible at \(K\).
Therefore \(K\Vdash\Gamma(P,\D)\) and \(K\not\Vdash e_r\),
so \(K\) refutes their implication.
\end{proof}

For the converse, we start with a valuation and a world \(W\)
such that \(W\Vdash\Gamma(P,\D)\) and \(W\not\Vdash e_r\).
We must assign roles to subsets of \(W\). Persistence makes every
clause available throughout \(F(W)\). For each nonempty
\(X\subseteq W\), put
\(G(X):=\{\lambda\in P:X\not\Vdash e_\lambda\}\).
We will recover roles from the sets \(G(X)\).

\begin{lemma}[Properties recovered from the clauses]
\label{lem:G}
For every nonempty \(X\subseteq W\), \(G(X)\) is a nonempty
upset.  If \(\varnothing\neq Y\subseteq X\), then
\(G(Y)\subseteq G(X)\).  For any \(U\subseteq P\), forcing
\(\bigwedge_{\lambda\notin U}e_\lambda\) at \(X\) is equivalent
to \(G(X)\subseteq U\).
\end{lemma}

\begin{proof}
Clause (N) prevents \(X\) from forcing all the point variables, so
\(G(X)\neq\varnothing\).  The (O) implications give upward closure
along each step from the root to a Mid-role and from a Mid-role to a
supported maximal role.  These steps generate the order, because every maximal
role appears in a support in the table defining \(P_D\).
Persistence of the variables gives \(G(Y)\subseteq G(X)\).
For the last assertion, \(X\) forces the conjunction iff
\(X\Vdash e_\lambda\) for every \(\lambda\notin U\), which is
equivalent to \(G(X)\subseteq U\).
\end{proof}

We can assign role \(\lambda\) to \(X\) when
\(G(X)=\up\lambda\). If \(G(X)\) has no least element,
\(X\) remains unassigned. The next lemma supplies subsets to
which we can assign the required roles.

\begin{lemma}[Principalization]
\label{lem:climb}
If \(\lambda\in G(X)\), some nonempty \(Y\subseteq X\) has
\(G(Y)=\up\lambda\).
\end{lemma}

\begin{proof}
If \(\lambda=r\), upward closure gives \(G(X)=P\), so we take
\(Y=X\).  Otherwise, (B) together with
\(X\not\Vdash e_\lambda\) gives
\(X\not\Vdash A_\lambda\to e_\lambda\).  By
\eqref{eq:failure-implication}, some nonempty \(Y\subseteq X\)
forces \(A_\lambda\) and fails \(e_\lambda\).
Thus \(\lambda\in G(Y)\subseteq\up\lambda\), and upward
closure makes the containment an equality.
\end{proof}

\begin{proposition}[Reverse recognition on the refuting world]
\label{prop:recognition-cone}
If \(W\Vdash\Gamma(P,\D)\) and \(W\not\Vdash e_r\), then
\((P,\D)\) has a realization on \(W\).
\end{proposition}

\begin{proof}
We assign a role precisely when the set of visible roles has a least
element:
\begin{equation}
 \role_W(X)=\lambda
 \quad\Longleftrightarrow\quad G(X)=\up\lambda.
\label{eq:old-4-29}
\end{equation}
Other worlds remain unassigned.  This defines a partial function,
since \(\up\lambda=\up\gamma\) entails both
\(\lambda\leq\gamma\) and \(\gamma\leq\lambda\), hence equality.

For (H0), the root is visible at \(W\), so \(G(W)=P=\up r\) and
\(\role_W(W)=r\). On a singleton \(\{k\}\), we choose any visible
\(\lambda\).  Principalization has only one nonempty subset to
choose, so \(G(\{k\})=\up\lambda\).  So the singleton receives
a role.  This role is maximal: if \(\gamma>\lambda\), then
\(\gamma\) is also visible, and principalization would give
\(G(\{k\})=\up\gamma\), contradicting
\(\lambda\notin\up\gamma\).  In particular, \(W\) itself
cannot be a singleton, since \(r\) is nonmaximal.

For (H1)(a), if \(\role_W(X)=\lambda\) and \(\lambda\leq\gamma\), then
\(\gamma\in G(X)\).  Principalization gives a nonempty
\(Y\subseteq X\) with \(G(Y)=\up\gamma\), hence role \(\gamma\).

For (H1)(b), suppose \(\role_W(X)=\lambda\),
\(\role_W(Y)=\gamma\), and \(Y\subseteq X\). Then
\(\up\gamma=G(Y)\subseteq G(X)=\up\lambda\).
Since \(\gamma\in\up\gamma\), we obtain \(\lambda\leq\gamma\).

For (H2), we take any demand \(\{\lambda,\gamma\}\in\D\)
and any sets \(X,Y\) with \(\role_W(X)=\lambda\) and
\(\role_W(Y)=\gamma\). We must find a nonempty subset of
\(X\cup Y\) whose assigned role lies outside the two input upsets.
Put \(U=X\cup Y\). Monotonicity gives
\(\up\lambda\cup\up\gamma\subseteq G(U)\). If equality held,
\(U\) would force \(A_{\lambda,\gamma}\) while failing both
\(e_\lambda\) and \(e_\gamma\), contrary to (S).
We therefore choose
\(\delta\in G(U)\setminus(\up\lambda\cup\up\gamma)\).
Lemma~\ref{lem:climb} gives a nonempty \(R\subseteq U\) with
\(G(R)=\up\delta\). By our assignment,
\(\role_W(R)=\delta\). Thus \(R\subseteq X\cup Y\) has the
required role and is a response to \(X,Y\).
\end{proof}

However, the formula may fail at a proper subset \(W\) of the
original carrier \(K\). What we have obtained is a realization on \(W\),
and still need one on \(K\). To obtain it, we map \(K\) onto
\(W\) and give each subset of \(K\) the role of its image,
whenever that image has a role.

\begin{lemma}[Lifting along a surjection]
\label{lem:carrier-lifting}
Let \(q:K\to W\) be a surjection of finite nonempty sets.
Given a realization \(\role_W\), define
\(\role_K(X):=\role_W(q[X])\), with the left side defined
iff the right side is.  This is a realization on \(K\).
\end{lemma}

\begin{proof}
Since \(q[K]=W\), the carrier receives the root role.  A singleton maps
to a singleton, so it receives a maximal role.  This proves (H0).

For (H1)(a), suppose \(\role_K(X)=\lambda\) and
\(\lambda\leq\gamma\). Since \(\role_W(q[X])=\lambda\),
(H1)(a) on \(W\) supplies a nonempty \(C\subseteq q[X]\)
with \(\role_W(C)=\gamma\). We put \(Y=X\cap q^{-1}[C]\).
Each \(c\in C\) has a preimage in \(X\), which belongs to \(Y\),
so \(q[Y]=C\). Thus \(Y\) is nonempty, \(Y\subseteq X\),
and \(\role_K(Y)=\gamma\).
For (H1)(b), \(Y\subseteq X\) gives \(q[Y]\subseteq q[X]\),
and (H1)(b) on \(W\) gives
\(\role_K(X)\leq\role_K(Y)\).

For (H2), take \(X,Y\) with \(\role_K(X)=\lambda\),
\(\role_K(Y)=\gamma\), and \(\{\lambda,\gamma\}\in\D\).
Their images have these same roles on \(W\), so we take
\(R=\Resp(q[X],q[Y])\) in that realization. Put
\(Z=(X\cup Y)\cap q^{-1}[R]\).
The same preimage argument gives \(q[Z]=R\). Hence \(Z\) is a
nonempty subset of \(X\cup Y\), and
\(\role_K(Z)=\role_W(R)\notin\up\lambda\cup\up\gamma\).
It is therefore a response to \(X,Y\) on \(K\).
\end{proof}

\begin{theorem}[Recognition on a fixed carrier]
\label{thm:recognition}
For every finite nonempty \(K\),
\begin{equation}
 F(K)\not\models\alpha(P,\D)
 \quad\Longleftrightarrow\quad
 (P,\D)\text{ has a realization on }K.
\label{eq:old-4-34}
\end{equation}
\end{theorem}

\begin{proof}
If \((P,\D)\) has a realization on \(K\),
Proposition~\ref{prop:recognition-forward} gives
\(F(K)\not\models\alpha(P,\D)\).
Conversely, frame invalidity supplies a valuation and, by
\eqref{eq:failure-implication}, a nonempty \(W\subseteq K\)
with \(W\Vdash\Gamma(P,\D)\) and \(W\not\Vdash e_r\).
Proposition~\ref{prop:recognition-cone} gives a realization on \(W\).
Choose \(w_0\in W\), and define \(q:K\to W\) by fixing
\(W\) and sending every element of \(K\setminus W\) to \(w_0\).
This map is surjective, so Lemma~\ref{lem:carrier-lifting}
gives a realization on \(K\).
\end{proof}

\begin{corollary}[Finite recognition]
\label{cor:recognition-finite}
The pair \((P,\D)\) has a finite realization iff some
finite Medvedev frame refutes \(\alpha(P,\D)\).  If finite
realizations exist, their carrier sizes are exactly the integers
\(n\geq n_0\), for some \(n_0\geq2\).
\end{corollary}

\begin{proof}
The equivalence follows from the theorem.  Every larger finite carrier
maps onto a given one, so lifting makes the possible sizes upward
closed.  Their least member is at least two by (H0), since the root
is nonmaximal.
\end{proof}

\section{The complexity of Medvedev logic}
\label{sec:assembly}

It remains to combine torus tilability, finite realizability, and
failure of the recognizing formula into a proof of \(\Pi^0_1\)-hardness.
Membership in \(\Pi^0_1\) follows by searching
\(F_1,F_2,\ldots\) for a countermodel to an input formula.

We use effective natural-number encodings of formulas, finite Wang
systems, and Turing machines. A set of codes is \emph{recursively
enumerable} (r.e.), or \emph{semidecidable}, if an algorithm halts
on exactly its members. It may run forever on a nonmember.
A set \(A\) belongs to \(\Pi^0_1\) when its complement is r.e.,
equivalently, when \(e\in A\) holds exactly if \(R(e,s)\) holds
for every \(s\in\mathbb N\), for some decidable relation \(R\).
A \emph{decision procedure} must halt on every input and correctly
determine membership.

A \emph{computable many-one reduction} from \(A\) to \(B\) is a
function \(f\), computed by an algorithm that halts on every input,
such that \(e\in A\) iff \(f(e)\in B\). A set is
\(\Pi^0_1\)-\emph{hard} if every \(\Pi^0_1\) set reduces to it,
and \(\Pi^0_1\)-\emph{complete} if it is also in \(\Pi^0_1\).
For the lower bound, it suffices to reduce nonhalting on empty input
to membership in \(\ML\), since nonhalting is
\(\Pi^0_1\)-complete.\footnote{Given \(e\), construct a machine
that tests \(R(e,0),R(e,1),\ldots\) and halts at the first failure.
It runs forever iff \(e\in A\).}

The reduction uses the following form of the periodic domino theorem.
The result is due to Gurevich and Koryakov \cite{GurevichKoryakov1972}.
We use Jeandel's construction \cite{Jeandel2010}, which yields it
directly in the form needed here.

\begin{theorem}[Effective periodic domino theorem]
\label{thm:import}
From an index of a Turing machine \(M\), one can compute a finite Wang
system \(D_M\) such that
\begin{equation}
 M\text{ halts on the empty input}
 \quad\Longleftrightarrow\quad
 D_M\text{ tiles a finite torus}.
\label{eq:old-5-2}
\end{equation}
\end{theorem}

\begin{proof}[Explanation of the formulation]
We first put \(M\) into the machine format used in the construction
and prepend two dummy transitions. This effective change preserves
halting and ensures that any halting computation has at least two
steps, as the construction requires. Jeandel's construction takes as
input an arbitrary aperiodic tileset, which we fix once and for
all,\footnote{An \emph{aperiodic} tileset tiles the plane but admits
no periodic tiling. Jeandel's construction requires this property of
its input and does not itself produce aperiodic tilesets. Any
aperiodic tileset may be used, and the choice does not affect the
theorem. Such tilesets exist by Berger's theorem \cite{Berger1966}.
Robinson \cite{Robinson1971} gave the first small explicit set, and
Kari \cite{Kari1996} and Culik \cite{Culik1996} gave aperiodic sets
of fourteen and thirteen Wang tiles. We take Culik's tiles. With the
tileset fixed, the construction is uniform in \(M\), so \(D_M\) is
computable from an index of \(M\).} and produces a finite tileset
admitting a periodic plane tiling iff the padded machine
halts. We convert it to \((T,H,V)\) by making \(H\) match east
colours to west colours, and \(V\) match north colours to south
colours. These relations are computable from the tiles. Jeandel calls
a plane tiling \(c\) \emph{periodic} when
\(c_{i,j}=c_{i+p,j}=c_{i,j+p}\) for some \(p\geq1\) and all
\(i,j\). Such a tiling descends to a torus tiling of
\(\mathbb Z_p\times\mathbb Z_p\). Conversely, a tiling of
\(\mathbb Z_m\times\mathbb Z_n\) extends to a periodic plane
tiling with \(p=\operatorname{lcm}(m,n)\).
\end{proof}

\begin{proposition}
\label{prop:upper}
The set \(\ML\) belongs to \(\Pi^0_1\).
\end{proposition}

\begin{proof}
Given a formula \(\varphi\), we search \(F_1,F_2,\ldots\) in order.
On each finite frame, we list its upsets by testing all sets of worlds.
For the finitely many variables occurring in \(\varphi\), all
assignments of these upsets form a finite, effectively enumerable list
of persistent valuations.

For each valuation and each world, we compute forcing recursively on
\(\varphi\). The implication clause is a finite check over the
nonempty subsets of that world. We halt if any world fails \(\varphi\).
Every stage of the search is finite, and the search halts iff
a finite Medvedev frame refutes \(\varphi\). It therefore semidecides
the complement of \(\ML\), as required.
\end{proof}

\begin{theorem}[Main theorem]
\label{thm:main}
Medvedev logic is \(\Pi^0_1\)-complete under computable many-one
reductions.
\end{theorem}

\begin{proof}
For a finite Wang system \(D\), we put
\(\alpha_D=\alpha(P_D,\D_D)\), using the recognizing formula of
Section~\ref{sec:recognition}. The transformation
\(D\mapsto\alpha_D\) is computable. Its steps list the roles, read
off their supports, list the demands, and write the corresponding
formula clauses. Each list is obtained by finite loops over \(T\),
\(T^2\setminus H\), and \(T^2\setminus V\). In the explicit
relational presentation of \(D\), the resulting formula has
polynomial size.

The preceding sections establish
\begin{equation}
\begin{aligned}
 D\text{ tiles a finite torus}
 &\quad\Longleftrightarrow\quad
 \WM(D)\text{ has a finite realization}\\
 &\quad\Longleftrightarrow\quad
 \alpha_D\notin\ML.
\end{aligned}
\label{eq:old-5-6}
\end{equation}
The first equivalence combines
Theorems~\ref{thm:realization-torus} and~\ref{thm:torus-realization}.
The second is Corollary~\ref{cor:recognition-finite} applied to
\(\WM(D)\).

Given \(M\), we compute \(D_M\) by Theorem~\ref{thm:import}
and then \(\alpha_{D_M}\). Combining that theorem with
\eqref{eq:old-5-6} and taking complements gives
\begin{equation}
 M\text{ does not halt on the empty input}
 \quad\Longleftrightarrow\quad
 \alpha_{D_M}\in\ML.
\label{eq:old-5-8}
\end{equation}
This computable many-one reduction proves \(\Pi^0_1\)-hardness.
Proposition~\ref{prop:upper} establishes membership in \(\Pi^0_1\).
\end{proof}

\begin{corollary}
\label{cor:nore}
Medvedev logic is undecidable and is not recursively enumerable. In
particular, it admits no recursively enumerable sound and complete
proof calculus.
\end{corollary}

\begin{proof}
A decision procedure for \(\ML\), composed with the reduction in
\eqref{eq:old-5-8}, would decide nonhalting, and \emph{eo ipso} the halting problem,
which is impossible. Now suppose \(\ML\) were r.e.
Its complement is already r.e. by Proposition~\ref{prop:upper}, so
we could run both membership searches in parallel. Exactly one would
halt, providing a decision procedure, again a contradiction. Finally,
a recursively enumerable sound and complete proof calculus would
enumerate exactly the formulas in \(\ML\), which has just been
excluded.
\end{proof}

\section*{Use of generative AI}

This work began as an attempt to establish the decidability of
Medvedev logic. The initial investigation relied heavily on generative
AI systems for proof search. During the dialogue and prompting I obtained numerous partial results and
simplifications, but these eventually suggested that the project should
proceed in the opposite direction, namely to prove undecidability. The initial construction, the proof sketches, and the idea of using tilings were obtained from generative AI systems. The construction was further honed in cooperation with AI systems. In that process I subsequently reconstructed, revised and clarified the preliminary results. Furthermore, with the help of generative AI I obtained a formalization of the paper in Lean up to Section~\ref{sec:assembly}, available at \url{https://github.com/Haptism89/medvedev-lean}. For the computability results of Section~\ref{sec:assembly}, I was unable, even with the help of AI, to reconstruct the halting problem for a machine model slightly different from the one used by default in the standard mathematical library of Lean. These results about the halting problem are standard, however, and play no role in the correctness of the paper.

For proof search, prompting and dialogue I used models such as Fable,
Sonnet, Opus, and ChatGPT~5.6 SOL. The same models also produced parts
of the \LaTeX{} source, including the diagrams, and were used to
improve the language throughout the paper.

This paper is the first part of a broader project in which I use
generative AI to address open problems in philosophical logic, in
particular questions of decidability and the theory of
non-deterministic matrix semantics. I have already obtained partial
results in both areas which I find promising, but they have not yet
been checked. Personally, I find this time quite interesting, given
the possibilities of interaction between human researchers and
generative AI. These possibilities seem to me practically limitless,
and I would greatly appreciate finding like-minded researchers to
help with the ideas mentioned above. Two kinds of help in particular
would be welcome. My knowledge of Lean
is basic, and the repository accompanying this paper is my first, so I
would be glad to cooperate with anyone experienced in formalization.
I would equally be glad to join forces with anyone interested in the
logical topics above, so that we can work on these questions together.
Please contact me at
\href{mailto:haptism89@gmail.com}{haptism89@gmail.com}.

\appendix
\section{Proofs of the opening lemmas and the response-role table}
\label{app:opening-proofs}

We use the notation of Section~\ref{sec:wmpair}, fixing an arbitrary
finite realization for the two lemma proofs.

\begingroup
\small
\setlength{\abovedisplayskip}{5pt plus 2pt minus 1pt}
\setlength{\belowdisplayskip}{5pt plus 2pt minus 1pt}
\setlength{\abovedisplayshortskip}{3pt plus 1pt}
\setlength{\belowdisplayshortskip}{3pt plus 1pt}

\phantomsection\label{proof:basic-concrete}
\begin{proof}[Proof of Lemma~\ref{lem:basic-concrete}]
For (a), choose \(k\in C\), and let \(\xi\) be the Max-role of
\(\{k\}\).  Since \(\{k\}\subseteq C\), condition (H1)(b) gives
\(\chi\leq\xi\).  Both roles are maximal, so \(\chi=\xi\).

For the first direction of (b), choose \(k\in A\), and let \(\chi\) be
the Max-role of \(\{k\}\).  The containment \(\{k\}\subseteq A\) and
(H1)(b) give \(\lambda\leq\chi\), hence
\(\chi\in\MaxOf{\lambda}\).  Conversely, fix
\(\chi\in\MaxOf{\lambda}\).  Since \(\lambda<\chi\), condition
(H1)(a) supplies a nonempty set \(C\subseteq A\) with
\(\role(C)=\chi\).  Part (a) shows that every element of \(C\) has
singleton role \(\chi\), so at least one such element occurs in \(A\).

For (c), if \(B\subseteq A\), condition (H1)(b) gives
\(\lambda\leq\gamma\), contrary to the hypothesis.
\end{proof}

\phantomsection\label{proof:response-test}
\begin{proof}[Proof of Lemma~\ref{lem:response-test}]
Every singleton role in \(X\cup Y\) lies in
\(\MaxOf{\lambda}\cup\MaxOf{\gamma}\), by
Lemma~\ref{lem:basic-concrete}(b).

For (a), inspection of (D1)--(D7) shows that the union of the input
supports always omits at least one of the two step keys
\(\keyc{\StepL{h}},\keyc{\StepL{v}}\). Call an omitted key
\(\kappa\). If \(Z\) had role \(r\), then (H1)(a) would give a
nonempty subset of \(Z\) with role \(\kappa\). By
Lemma~\ref{lem:basic-concrete}(a), that subset would contain elements
of singleton role \(\kappa\), which are absent from \(X\cup Y\).
Thus \(\delta\ne r\). If \(\delta\) were a Max-role, the same
lemma would make every element of \(Z\) have singleton role \(\delta\).
Hence \(\delta\) would belong to an input support and therefore to
\(\up\lambda\cup\up\gamma\), contrary to legality. Only Mid-roles remain.

For (b), every singleton role in \(\MaxOf{\delta}\) occurs in \(Z\), by
Lemma~\ref{lem:basic-concrete}(b). Since \(Z\subseteq X\cup Y\),
it belongs to one of the input supports.

For (c), suppose \(\MaxOf{\delta}\) misses
\(\MaxOf{\lambda}\). Then \(Z\cap X=\varnothing\), because an
element in the intersection would have a singleton role in both supports.
Consequently \(Z\subseteq Y\). By (H1)(b),
\(\gamma\leq\delta\), contradicting legality. Exchanging \(X\)
and \(Y\) proves the other intersection is nonempty.
\end{proof}

\phantomsection\label{proof:forced-table}
\begin{proof}[Proof of Proposition~\ref{prop:forced-table}]
For an input pair \(\{\lambda,\gamma\}\), put
\(U=\MaxOf{\lambda}\cup\MaxOf{\gamma}\) and
\(B(U)=\{\delta\in\Mid:\MaxOf{\delta}\subseteq U\}\).
Thus \(B(U)\) lists all Mid-roles whose required singleton roles are
available in the input supports. A role \(\delta\in B(U)\) is
admissible iff
\[
 \delta\notin\{\lambda,\gamma\},\qquad
 \MaxOf{\delta}\cap\MaxOf{\lambda}\ne\varnothing,\qquad
 \MaxOf{\delta}\cap\MaxOf{\gamma}\ne\varnothing.
\]
The first condition is legality, since distinct Mid-roles are
incomparable. We compute \(B(U)\) from the support table, then apply
these three conditions. All private keys are distinct, including keys
for different roles with the same tile index.

\emph{(M1) and (M2).}
In both cases \(U=\{\sa{h},\nx{h},\tg{h}\}\), and
\(B(U)=\allowbreak\{\CPos,\allowbreak\CSep,\allowbreak\NFL{h},\allowbreak\SFL{h}\}\).
For (M1), legality removes \(\CPos\) and \(\NFL{h}\).
The role \(\SFL{h}\) misses the support
\(\{\nx{h}\}\) of the second input, whereas \(\CSep\) meets
both input supports. Hence only \(\CSep\) is admissible.
For (M2), legality removes \(\CSep\) and \(\SFL{h}\).
The role \(\NFL{h}\) misses the support
\(\{\sa{h}\}\) of the second input, whereas \(\CPos\) meets
both input supports. Hence only \(\CPos\) is admissible.

\emph{(M5).}
Here \(U\) consists of the four Max-roles above a cell role and the two
tags. Since it contains no private key, the complete candidate set is
\[
\begin{split}
 B(U)=\{&\CPos,\CSep,\NFL{h},\SFL{h},\\
         &\RPos,\RSep,\NFL{v},\SFL{v}\}
          \cup\{\CellL{s}:s\in T\}.
\end{split}
\]
Each of the first four roles misses the row input support, and each of
the next four misses the column input support. Every cell role is
distinct from the inputs. Its support meets the column support in
\(\{\sa{h},\nx{h}\}\) and the row support in
\(\{\sa{v},\nx{v}\}\). Thus exactly the cell roles are admissible.

\emph{(M6).}
Fix one of the four projection roles \(G_t\). The union \(U\)
consists of the four Max-roles above a cell role and \(\keyc{G_t}\), so
\[
 B(U)=\{\NFL{h},\SFL{h},\NFL{v},\SFL{v},\KeyL{G_t},G_t\}
       \cup\{\CellL{s}:s\in T\}.
\]
Indeed, a generator requires a tag, and every other keyed role
requires an unavailable key. In particular, if \(G_t\) is a source,
its expected-neighbour role also requires a step key.
Meeting the provider support \(\{\keyc{G_t}\}\) leaves only
\(\KeyL{G_t}\) and \(G_t\). Legality removes the former.
The latter is distinct from both inputs, and its support meets the
cell support in two Max-roles. Hence only \(G_t\) is admissible.

\emph{(M7).}
Here \(U=\{\sa{h},\nx{h},\tg{h},\keyc{\StepL{h}}\}\), and
\[
 B(U)=\{\CPos,\CSep,\NFL{h},\SFL{h},
          \KeyL{\StepL{h}},\StepL{h}\}.
\]
An expected-neighbour role is excluded because its source key is
absent. Meeting the provider support leaves only
\(\KeyL{\StepL{h}}\) and \(\StepL{h}\); legality removes the
provider. The step role meets the separator support in
\(\{\sa{h},\nx{h}\}\), so it is admissible and unique.

\emph{(M9).}
The input union is
\(U=\{\sa{h},\sa{v},\nx{h},\keyc{\SrcL{h}{t}},
\keyc{\StepL{h}}\}\).
There is no tag or \(\nx{v}\), and the only available keys are
those of this source and this step. The support table therefore gives
\[
\begin{split}
 B(U)=\{&\SFL{h},\NFL{h},\SFL{v},
         \SrcL{h}{t},\StepL{h},\ExpL{h}{t},\\
       &\KeyL{\SrcL{h}{t}},\KeyL{\StepL{h}}\}.
\end{split}
\]
Legality removes the source and step inputs. The roles
\(\NFL{h}\) and \(\KeyL{\StepL{h}}\) miss the source support;
\(\SFL{v}\) and \(\KeyL{\SrcL{h}{t}}\) miss the step support.
The two remaining roles pass both intersection tests:
\(\SFL{h}\) uses the common Max-point \(\sa{h}\), while
\(\ExpL{h}{t}\) meets the source support in
\(\{\sa{v},\keyc{\SrcL{h}{t}}\}\) and the step support in
\(\{\nx{h},\keyc{\StepL{h}}\}\).
Thus precisely \(\SFL{h}\) and \(\ExpL{h}{t}\) are admissible.

\emph{(M11).}
For any \(t,u\in T\), the supports of
\(\ExpL{h}{t}\) and \(\NbrL{h}{u}\) have union
\(U=\{\nx{h},\sa{v},\keyc{\SrcL{h}{t}},
\keyc{\StepL{h}},\keyc{\NbrL{h}{u}}\}\).
Since \(\sa{h}\), \(\nx{v}\), and both tags are absent,
\[
\begin{split}
 B(U)=\{&\NFL{h},\SFL{v},\ExpL{h}{t},\NbrL{h}{u},\\
        &\KeyL{\SrcL{h}{t}},\KeyL{\StepL{h}},
          \KeyL{\NbrL{h}{u}}\}.
\end{split}
\]
Legality removes the two inputs. The source-key and step-key
providers miss the candidate support, and the candidate-key provider
misses the expected support. The remaining roles \(\NFL{h}\)
and \(\SFL{v}\) have supports \(\{\nx{h}\}\) and
\(\{\sa{v}\}\), respectively, both contained in each input
support. These are exactly the admissible roles. When
\((t,u)\notin H\), the input pair is a demand, giving (M11).

\emph{(M3), (M4), (M8), (M10), and (M12).}
Apply the axis exchange \(\sigma\) defined in Section~\ref{sec:wmpair} to the calculations
for (M1), (M2), (M7), (M9), and (M11), respectively.
As established there, \(\sigma\) preserves legality, support
containment, and both intersection tests.
The resulting admissible sets are, respectively,
\(\{\RSep\}\), \(\{\RPos\}\), \(\{\StepL{v}\}\),
\(\{\SFL{v},\ExpL{v}{t}\}\), and
\(\{\NFL{v},\SFL{h}\}\).
The last support calculation holds for every \(t,u\); it applies
to the demand in (M12) when \((t,u)\notin V\).
No equality between \(H\) and \(V\) is needed.

We have thus computed exactly the admissible roles in every row.
Lemma~\ref{lem:response-test} places the role of every actual response
among them.
\end{proof}
\endgroup

\printbibliography
\end{document}